\documentclass[a4paper,onecolumn,11pt,accepted=2025-11-07]{quantumarticle}
\pdfoutput=1
\usepackage[utf8]{inputenc}
\usepackage[english]{babel}
\usepackage[T1]{fontenc}
\usepackage{amsmath}
\usepackage{hyperref}
\usepackage[numbers,sort&compress]{natbib}
\usepackage{lipsum}

\usepackage{graphicx} 
\usepackage{epstopdf}
\usepackage{epsfig}
\usepackage{amsmath,amsthm, amsfonts,amssymb} 

\usepackage{rotating}
\usepackage{verbatim}
\usepackage{longtable}
\usepackage{lscape}
\usepackage{subfig}
\usepackage{float}
\usepackage{color,soul,mathtools} 
\usepackage[dvipsnames]{xcolor} 
\usepackage{braket} 
\usepackage{indentfirst} 
\usepackage{bm} 
\usepackage{mathrsfs} 
\usepackage{latexsym} 

\usepackage{hyperref} 
\hypersetup{
	colorlinks = true,
	linkcolor = blue,
	urlcolor = blue,
	citecolor = blue,
	anchorcolor = black,
}

\usepackage{microtype} 
\usepackage{bookmark} 
\usepackage{tikz}[2010/10/13] 

\usepackage{mathtools} 
\usepackage{ifpdf} 
\usepackage{doi} 

\usepackage{url}
\usepackage{rotating}
\usepackage[innercaption]{sidecap}
\usepackage{pdfpages}
\usepackage{xspace}

\newtheorem{Thm}{Theorem}[section]

\newtheorem{proposition}[Thm]{\bf Proposition}
\newtheorem{corollary}[Thm]{\bf Corollary}
\newtheorem{lemma}[Thm]{\bf Lemma}

\newtheorem{remark}[Thm]{Remark}
\newtheorem{definition}[Thm]{\bf Definition}

\newtheorem{theorem}[Thm]{\bf Theorem}

\newcommand{\GG}[1]{}

\tikzstyle WL=[line width=5pt,opacity=1.0]

\newcommand{\mn}{/-3}
\newcommand{\nn}{/3}
\newcommand{\size}[1]{\fontsize{9pt}{\baselineskip}\selectfont{#1}}

\newcommand{\fmeasure}[5]{
	\node at (#1-#3,#2-#4) {\size{$#5$}};
	\draw (#1,#2) --(#1,#2-#4)  arc (0:180:#3) -- (#1-#3-#3,#2);
}

\newcommand{\fqudit}[5]{
	\node at (#1-#3,#2+#4) {\size{$#5$}};
	\draw (#1,#2) --(#1,#2+#4)  arc (0:-180:#3) -- (#1-#3-#3,#2); 
}

\newcommand{\fbraid}[4]{
	\draw (#1,#2)--(#3,#4);  
	\draw (#1,#4)--(2/3*#1+1/3*#3,2/3*#4+1/3*#2); 
	\draw (#3,#2)--(2/3*#3+1/3*#1,2/3*#2+1/3*#4); 
}

\begin{document}

\title{A Graphical Calculus for Quantum Computing with Multiple Qudits using Generalized Clifford Algebras}

\author{Robert Lin}
\affiliation{Department of Physics, Harvard University, Cambridge, MA 02138}
\email{robertlin@g.harvard.edu}
\homepage{https://sites.google.com/view/robert-lin-mathphys-appliedmat}
\orcid{0000-0003-0290-4698}
\thanks{current affiliation is School of Engineering and Applied Sciences, Harvard University, Cambridge, MA 02138; and Department of Chemistry and Chemical Biology, Harvard University, Cambridge, MA 02138.}

\maketitle

\tableofcontents

\begin{abstract}

In this work, we develop a graphical calculus for multi-qudit computations with generalized Clifford algebras, building off the algebraic framework developed in \cite{Lin1}. We build our graphical calculus out of a fixed set of graphical primitives defined by algebraic expressions constructed out of elements of a given generalized Clifford algebra, a graphical primitive corresponding to the ground state, and also graphical primitives corresponding to projections onto the ground state of each qudit. We establish many properties of the graphical calculus using purely algebraic methods, including a novel algebraic proof of a Yang-Baxter equation and a construction of a corresponding braid group representation. Our algebraic proof, which applies to arbitrary qudit dimension, also enables a resolution of an open problem in \cite{Cobanera2014} on the construction of self-dual braid group representations for even qudit dimension.  We also derive several new identities for the braid elements, which are key to our proofs. Furthermore, we demonstrate that in many cases, the verification of involved vector identities can be reduced to the combinatorial application of two basic vector identities. Additionally, in terms of quantum computation, we demonstrate that it is feasible to envision implementing the braid operators for quantum computation, by showing that they are 2-local operators. In fact, these braid elements are \textit{almost} Clifford gates, for they normalize the generalized Pauli group up to an extra factor $\zeta$, which is an appropriate square root of a primitive root of unity.
\end{abstract}

\section{Introduction}

 Qudits are $d$-dimensional Hilbert spaces. Since when $d=2$, the qudits are just called qubits, qudits include qubits as a special case. By elementary scaling considerations, a system of $n$ qudits possesses $d^n$ states; if $d = 3$, $3^n$ is already exponentially greater than $2^n$. Thus, multi-qudit systems offer tremendous computational advantages, if implemented. A flurry of recent work (see the recent survey article \cite{Kiktenko}) has made platforms based on neutral atoms, trapped ions, superconducting platform, and photonics for qudits with $d>2$, as opposed to qubits, viable for performing quantum computation. 

Work by Kauffman (see \cite{Kauffman}) studied the usefulness of \textit{Clifford} algebras, which are noncommutative structures that can be constructed explicitly out of the Pauli operators $X_i$, $Z_i$, for quantum computation on qubits. In particular, Kauffman establishes the form of 2-local operators for the qubit case which satisfy a Yang-Baxter equation, i.e. one of the form $\sigma_{12} \sigma_{23} \sigma_{12} = \sigma_{23} \sigma_{12} \sigma_{23}$  (for short, $ABA = BAB$), where the indices on the $\sigma$'s indicate the pairs of qubits being entangled. 

The possibility of using \textit{generalized} Clifford algebras to represent operations with \textit{qudits} was indicated in the work of Cobanera and Ortiz \cite{Cobanera2014}. Crucially, \cite{Cobanera2014} indicated the particular relevance of operators of a particular type, self-dual braid group representations, to topological quantum processing. While \cite{Cobanera2014} was able to construct self-dual braid group representations for odd qudit dimension, the even case for $d>2$ was left as an open problem. Importantly, the bottleneck in their construction was a reliance on earlier results of Goldschmidt and Jones \cite{GJ1989}, which only applied to the odd case. When $d=2$, the solution is explicitly computable; for larger $d$, one has to work with an equation which has order $d^3$ terms on both sides, and the resulting equation is cubic in the coefficients to be solved.

A breakthrough on this problem was initiated by the work of Jaffe and Liu \cite{Jaffe2017}. In \cite{Jaffe2017}, the authors extend the work of Jones on planar algebras \cite{Jones} by considering a new structure which they called planar para algebras. Planar algebras are a diagrammatic axiomatization of a completely algebraic structure, known as the standard invariant for subfactors \cite{Jones}. The term subfactor refers to a factor within a factor, and a factor is a
unital $*$-algebra of bounded linear operators on a Hilbert space, with trivial
center and closed in the topology of pointwise convergence \cite{Jones}. 
In the work of Jaffe and Liu \cite{Jaffe2017}, a particular kind of unitary operator was defined for each qudit dimension, which the authors then used their diagrammatic theory to show that a set of these operators (at fixed qudit dimension) satisfies a set of Yang-Baxter braiding relations $ABA = BAB$ regardless of the qudit dimension.\footnote{We note that this operator reduces to the braiding element in Kauffman's paper \cite{Kauffman} when $d=2$.}

In this work, we provide an alternative route toward the self-dual braid group representations which depends only on the properties of a fixed generalized Clifford algebra. No subfactor theory, planar algebraic or planar para algebraic framework, or tensor category theory  (which doesn't apply here, since there is no global tensor product) is invoked. Hence, all the proofs are elementary, in the sense that they only depend on relatively complicated manipulations of large trigonometric sums. 
It is worth noting that the most important ingredient in the proofs is the use of the fact that the algebras have trivial center. The repeated application of this fact, combined with certain symmetries of the operators of \cite{Jaffe2017}, enables us (after some additional technical identities)  to establish a novel algebraic proof (requiring only the properties of the generalized Clifford algebra) that the operators of \cite{Jaffe2017} satisfy a Yang-Baxter braiding relation.

Our algebraic approach enables a proof of a much stronger result: the Yang-Baxter relation $ABA=BAB$ holds at the pairwise level, i.e. it depends only on the description of $A$ and $B$ in terms of a triple of generalized Clifford algebra generators $c_k$, $c_l$, and $c_m$. This result indicates that the braiding properties are inherent to the structure of the generalized Clifford algebra, and need not rely on a topological framework.

Putting all these new results together, we are thus in a position to formulate a 
graphical calculus for generalized Clifford algebras, which is laid out in Section 2 at the level of the diagrammatic representation, and developed at the level of the \textit{algebra} in Section 3. For our purposes, a graphical calculus is any set of diagrammatic replacement rules (i.e. rules for replacing one diagram by another) which are mutually compatible. The demonstration of mutual compatibility is \textit{a priori} a subtle task, since it involves the interpretation of diagrams.  For logical consistency, the reader should consider the graphical calculus as a transcription of algebraic identities into diagrammatic replacement rules. Thus, mutual compatibility is assured, 
as true statements are always compatible with each other. 

In terms of the graphical representation, the diagrams allowed  are a much smaller subset than as those of \cite{Jaffe2017}, in order to ensure \textit{unambiguous} identification of a graphical diagram (via vertical decomposition) with an algebraic expression. In line with the requisite of unambiguity of graphical-to-algebraic correspondence, no independent interpretation is made of the subcomponents of the diagrams. The latter constraint imposed by our work makes it necessary to specify in advance all the possible configurations one may encounter in a full diagram, and the corresponding algebraic expressions. This specification is accomplished using the tool of diagrammatic composition, originating from the theory of Temperley-Lieb algebras \cite{Temperley1971}, applied to a particular (small) set of graphical primitives which are specified in their completeness.

The other half of the picture, how to further extend the graphical calculus to  multi-qudit vector states, is tackled in Section 4, and depends on the particular representation of the generalized Clifford algebra being considered. Axioms to handle the required properties of the representation were introduced in our prior work \cite{Lin1}, and form the basis of this approach. 

Let us note that there are a number of important conceptual differences between our approach and the famous ZX calculus (qubit and qudit versions, see \cite{Poor2023} for a nice summary of the qudit case): the ZX calculi take their inspiration from category theory, as seen from the original work of Coecke and Duncan \cite{Coecke_2011}, which we do not use at all. Additionally, we try to stay away from generators and relations, because this \textit{a priori} leads one to (what seem to the author to be difficult) consistency checks, which are more easily resolved using algebraic axioms that lead one back to concrete matrix representations. That being said, the recent work on qudit ZX calculi converges with the present work (which originally appeared as a preprint in March 2021), in  that many (though not all) of the different flavors of results need to specialize to the different bases $N$, e.g., $N$ a power of an odd prime, or $N$ odd-prime, etc. (see \cite{Poor2023} for a comprehensive discussion). Perhaps not coincidentally, what stymied Cobanera and Ortiz \cite{Cobanera2014} in their quest to find self-dual unitary braid representations for \textit{all} $N$ was their dependence on the work of \cite{GJ1989}, which meant that they had to rely on a result which only applied to powers of odd primes. While it is difficult to engage in a direct comparison, due to the subtleties in the ZX calculi regarding soundness and completeness, there is a sense in which the problems which the present work addresses are not completely disconnected from the broader effort to render multi-qudit computation ``graphical.''

\section{Graphical Calculus: Diagrammatic Setup}
\label{graph}

\subsection{Building Blocks}
The philosophy followed in the graphical calculus we present is that the diagrams drawn are \textbf{indivisible}. No a priori meaning is assigned to the subcomponents of the diagrams, i.e. a single strand, or a single cap, or a single cup.  The philosophy adopted is that the algebraic framework of \cite{Lin1} ought to be robust enough that one can \textbf{derive} a posteriori a large number of algebraic relations, and therefore by proving more and more relations, the initially content-free diagrams acquire new, emergent properties. On a technical level, this approach leads to a more basic construction of a graphical calculus which is directly built out of the elements of the generalized Clifford algebra, which is justified by the axiomatic framework.

In devising the graphical representation, we need to consider at the outset what kind of diagrams should be allowed. From the perspective of mathematical rigor, if one proceeds on entirely algebraic grounds, and it is decided to base the manipulation of graphical diagrams on corresponding algebraic identities, it becomes necessary that each graphical diagram have a \textit{unique} algebraic expression. Note that the word ``expression'' is used, as opposed to ``value.'' Two expressions may evaluate to the same algebraic element in the generalized Clifford algebra. Likewise, two graphical diagrams may be \textit{different} in the sense that they correspond to different algebraic expressions, but \textit{equal} in the sense that the expressions they correspond to can be shown to be algebraically equal (under the relations of the generalized Clifford algebra and two additional representation-theoretic axioms).

To be mathematically precise, one has to specify in what sense one means ``uniqueness.'' In this article, by uniqueness of the algebraic expression corresponding to a diagram, it is meant that the formal algebraic expression (forgetting all properties of the generalized Clifford algebra, \textit{except} associativity, the property that $a(bc) = (ab)c$ for any elements $a,b,c$ of the algebra) obtained from the diagram is invariant under vertical decomposition of the diagram, \textit{up to} associativity. Thus, the graphical primitives are carefully chosen to guarantee uniqueness of an operator correspondence beyond diagrams and equations, a correspondence which is compatible with the vertical decomposition of diagrams. Adhering to this dictum results in a set of allowed diagrams that is much smaller than that of \cite{Jaffe2017}.

 \begin{definition}
     Fix $N$ a positive integer greater than 1, $n$ a positive integer at least 1, and consider the \textbf{generalized Clifford algebra}\footnote{The earliest paper introducing generalized Clifford algebras appears to be \cite{Morinaga} in 1952. Other early work included \cite{Yamazaki} in 1964, \cite{Popovici} in 1966, and \cite{Morris} in 1967.}  $\mathcal{C}_{2n}^{(N)}$ generated by (i.e. the smallest $\mathbb{C}$-algebra, closed under multiplication and addition, containing) $c_1$, $c_2$, $c_3$, $\ldots$ , $c_{2n}$ over the complex numbers, subject to $c_i c_j = q c_j c_i$ if $i<j$, and $c_i^N=1$ for all $i$. Here, $q=\exp(2\pi i/N)$ is a primitive $N$th root of unity. When $N=2$, one recovers the Clifford algebra with $2n$ generators.
 \end{definition} For our purposes, we will also need to define $\zeta$ satisfying $\zeta^2 = q$ and $\zeta^{N^2}=1$ according to the following lemma.
 \begin{lemma}
	\label{zeta}
	Let $q=\exp(2\pi i/N)$. If $N$ is odd, $\zeta=-\exp(\pi i/N)$ is the only square root of $q$ satisfying $\zeta^{N^2}=1$. If $N$ is even, setting $\zeta$ to be either square root of $q$ will satisfy $\zeta^{N^2}=1$.
\end{lemma}

 Let us first define a series of \textbf{graphical primitives}. These graphical primitives are the only allowed graphical elements in our graphical representation. Any diagram encoded using this set of graphical primitives must be specified by a sequence of graphical primitives. One may think of each diagram as a hieroglyph in an alphabet of hieroglyphs, and the sequence of hieroglyph as running from top to bottom. (This corresponds to the composition of operators, in which, in terms of the corresponding algebraic objects, the corresponding algebraic expression are given by a sequence of operations running from right to left.)

Fix $\delta=\sqrt{N}>0$. The following graphical primitives are defined in terms of the distinguished ground state (satisfying the two axioms) via:
\begin{definition}
 \begin{equation}\raisebox{-.15cm}{
 	\tikz{
 		\fqudit{0\mn}{-2.85\nn}{1\mn}{0.5\nn}{}
 }}\;
\raisebox{-.15cm}{
	\tikz{
		\fqudit{0\mn}{-2.85\nn}{1\mn}{0.5\nn}{}
}}\;\cdot \cdot
\raisebox{-.15cm}{
	\tikz{
		\fqudit{0\mn}{-2.85\nn}{1\mn}{0.5\nn}{}
}}\;
:=\delta^{n/2}\ket{\Omega}^{\otimes n}
\end{equation}

\begin{equation}
\raisebox{-.15cm}{
	\tikz{
		\fmeasure {0}{3.5\nn}{1\mn}{.5\nn}{}
}}\;
\raisebox{-.15cm}{
	\tikz{
		\fmeasure {0}{3.5\nn}{1\mn}{.5\nn}{}
}}\;
\cdot \cdot 
\raisebox{-.15cm}{
	\tikz{
		\fmeasure {0}{3.5\nn}{1\mn}{.5\nn}{}
}}\;
:= \delta^{n/2}
\bra{\Omega}^{\otimes n}
\end{equation}
\end{definition}

\begin{definition}
\begin{equation}
\raisebox{-.3cm}{	\tikz{
		\draw (1\mn,0)--(1\mn,1);
		\draw (0\mn,0)--(0\mn,1);
}} \;
\cdot \cdot \! 
\raisebox{-.3cm}{	\tikz{
		\draw (1\mn,0)--(1\mn,1);
		\draw (0\mn,0)--(0\mn,1);
		\node at (1.5\mn,3/4) {\size{$a$}};
}} \;
\cdot \cdot
\raisebox{-.3cm}{	\tikz{
		\draw (1\mn,0)--(1\mn,1);
		\draw (0\mn,0)--(0\mn,1);
	}} \;
:=c_{2k-1}^a
\end{equation}

\begin{equation}
    \raisebox{-.3cm}{	\tikz{
		\draw (1\mn,0)--(1\mn,1);
		\draw (0\mn,0)--(0\mn,1);
}} \;
\cdot \cdot \;
\raisebox{-.3cm}{	\tikz{
		\draw (1.5\mn,0)--(1.5\mn,1);
		\draw (0\mn,0)--(0\mn,1);
		\node at (0.5\mn,3/4) {\size{$b$}};
}} \;
\cdot \cdot 
\raisebox{-.3cm}{	\tikz{
		\draw (1\mn,0)--(1\mn,1);
		\draw (0\mn,0)--(0\mn,1);
}} \;
:=c_{2k}^b \;\;\;
\end{equation}
$
\forall a,b \in \mathbb{Z}.
$
Here we mean for the label $a$ to be placed immediately left of the $2k-1$-th strand, and the label $b$ to be placed immediately left of the $2k$-th strand. There are $2n$ total strands in each diagram.

We also define for completion that
\begin{equation}
    \raisebox{-.3cm}{	\tikz{
		\draw (1\mn,0)--(1\mn,1);
		\draw (0\mn,0)--(0\mn,1);
}} \;
\cdot \cdot 
\raisebox{-.3cm}{	\tikz{
		\draw (1\mn,0)--(1\mn,1);
		\draw (0\mn,0)--(0\mn,1);
}} \;
\cdot \cdot 
\raisebox{-.3cm}{	\tikz{
		\draw (1\mn,0)--(1\mn,1);
		\draw (0\mn,0)--(0\mn,1);
}} \;
:=1 \qquad \;
\end{equation}

Note that the identity primitive composed with itself ``is'' itself, graphically, which is consistent with its definition as being equal to 1. Similarly, the identity primitive composed (in either order) with the primitives for the powers of the generators $c_{k}$ again yields those same primitives. In this sense, the diagrammatic definitions are well-behaved.

\end{definition}

\begin{definition}
	\begin{equation}
 \raisebox{-.3cm}{	\tikz{
			\draw (1.5\mn,0)--(1.5\mn,1);
			\draw (0\mn,0)--(0\mn,1);
	}} \;
	\cdot \cdot
\raisebox{-.5cm}{
		\tikz{
			\fqudit {0}{0\nn}{1\mn}{.5\nn}{\phantom{ll}}
			\fmeasure {0}{3.5\nn}{1\mn}{.5\nn}{}
	}}\;
	\cdot \cdot
	\raisebox{-.3cm}{	\tikz{
			\draw (1.5\mn,0)--(1.5\mn,1);
			\draw (0\mn,0)--(0\mn,1);
	}} \;
	:=\delta E_k
 \end{equation}
	
	Here we mean for the ``cup-cap'' combination to be replacing the $2k-1$ and $2k$th strands.\footnote{In this respect, in our graphical calculus, we do not allow for the cup-cap combination which is prescribed in \cite{Jaffe2017}, i.e. we don't allow not-in-place placement, i.e. on the $2k$ and $(2k+1)$th strands, which loosely speaking, straddles different qudits. }
	There are $2n$ strands in total. 
	
\end{definition}

\begin{definition}
	We also define a graphical primitive, which we call the positive braid on strands $l$ and $l+1$, for $l=1,2,\ldots, 2n-1$:
	 \begin{equation}\raisebox{-.3cm}{
			\tikz{
				\fbraid{-1}{1}{0}{0}
				}}\;
		\raisebox{-.3cm}{	\tikz{
					\draw (1\mn,0)--(1\mn,1);
					\draw (0\mn,0)--(0\mn,1);
			}} \;
			\cdot \cdot 
			\raisebox{-.3cm}{	\tikz{
					\draw (1\mn,0)--(1\mn,1);
					\draw (0\mn,0)--(0\mn,1);
				}} \;
			:=b_{12} \;
	\end{equation}
		\begin{equation}
		 	\raisebox{-.3cm}{	\tikz{
		 		\draw (1\mn,0)--(1\mn,1);
		 		 }} \;
	 \raisebox{-.3cm}{
				\tikz{
					\fbraid{-1}{1}{0}{0}
			}} \;
			\raisebox{-.3cm}{	\tikz{
			\draw (0\mn,0)--(0\mn,1);
			}} \;
			\cdot \cdot 
			\raisebox{-.3cm}{	\tikz{
					\draw (1\mn,0)--(1\mn,1);
					\draw (0\mn,0)--(0\mn,1);
				}} \;
			:=b_{23} \;
			\end{equation}
			\begin{equation} \qquad \;
				 \raisebox{-.3cm}{	\tikz{
				 		\draw (1\mn,0)--(1\mn,1);
				 		\draw (0\mn,0)--(0\mn,1);
				 }} \;
				 \cdot \cdot 
				 \raisebox{-.3cm}{
				\tikz{
					\fbraid{-1}{1}{0}{0}
			}}\; \cdot \cdot 		
			\raisebox{-.3cm}{	\tikz{
					\draw (1\mn,0)--(1\mn,1);
					\draw (0\mn,0)--(0\mn,1);
			}} \;
			:=b_{k,k+1} \;
			\end{equation}
				\begin{equation} \quad \; \, \,
			\raisebox{-.3cm}{	\tikz{
					\draw (1\mn,0)--(1\mn,1);
					\draw (0\mn,0)--(0\mn,1);
			}} \;
			\cdot \cdot \cdot \cdot \;
		 \raisebox{-0.3cm}{\tikz	{
				\fbraid{-1}{1}{0}{0}
		}} \; 
			:=b_{2n-1,2n} \, \,
			\end{equation}
			which defines $2n-1$ different braid operators. 
			
			We also define graphical primitives for the corresponding negative braids:
			
			 \begin{equation}\raisebox{-.3cm}{
				\tikz{
					\fbraid{1}{1}{0}{0}
			}}\;
			\raisebox{-.3cm}{	\tikz{
					\draw (1\mn,0)--(1\mn,1);
					\draw (0\mn,0)--(0\mn,1);
			}} \;
			\cdot \cdot 
			\raisebox{-.3cm}{	\tikz{
					\draw (1\mn,0)--(1\mn,1);
					\draw (0\mn,0)--(0\mn,1);
			}} \;
			:=b_{21} \;
			\end{equation}
			\begin{equation}
			\raisebox{-.3cm}{	\tikz{
					\draw (1\mn,0)--(1\mn,1);
			}} \;
			\raisebox{-.3cm}{
				\tikz{
					\fbraid{1}{1}{0}{0}
			}} \;
			\raisebox{-.3cm}{	\tikz{
					\draw (0\mn,0)--(0\mn,1);
			}} \;
			\cdot \cdot 
			\raisebox{-.3cm}{	\tikz{
					\draw (1\mn,0)--(1\mn,1);
					\draw (0\mn,0)--(0\mn,1);
			}} \;
			:=b_{32} \;
			\end{equation}
			\begin{equation} \qquad \;
			\raisebox{-.3cm}{	\tikz{
					\draw (1\mn,0)--(1\mn,1);
					\draw (0\mn,0)--(0\mn,1);
			}} \;
			\cdot \cdot 
			\raisebox{-.3cm}{
				\tikz{
					\fbraid{1}{1}{0}{0}
			}}\; \cdot \cdot 		
			\raisebox{-.3cm}{	\tikz{
					\draw (1\mn,0)--(1\mn,1);
					\draw (0\mn,0)--(0\mn,1);
			}} \;
			:=b_{k+1,k} \;
			\end{equation}
			\begin{equation} \quad \; \, \,
			\raisebox{-.3cm}{	\tikz{
					\draw (1\mn,0)--(1\mn,1);
					\draw (0\mn,0)--(0\mn,1);
			}} \;
			\cdot \cdot \cdot \cdot \;
			\raisebox{-0.3cm}{\tikz	{
					\fbraid{1}{1}{0}{0}
			}} \; 
			:=b_{2n,2n-1} .\, \,
			\end{equation}
			
The algebraic definition of these braid elements\footnote{The special case in which $k$ and $l$ are adjacent was studied by Jaffe and Liu \cite{Jaffe2017}, which, to the best of the author's knowledge, is the first work to introduce this particular summation definition for the generalized Clifford algebra. A related summation expression for constructing a braid element is given by the work of Jones \cite{Jones1989} in the case that $N$ is a power of an odd prime.} is given by \begin{equation}b_{kl}:=\frac{\omega^{1/2}}{\sqrt{N}} \sum_{i=0}^{N-1} c_k^i c_l^{-i} \end{equation}  and \begin{equation} b_{lk}:=\frac{\omega^{-1/2}}{\sqrt{N}} \sum_{i=0}^{N-1} c_{l}^i c_{k}^{-i} \end{equation}
for $k<l$ in $\{1,2,\ldots,2n\}$. Here, 
\begin{equation}
\omega:=\frac{1}{\sqrt{N}} \sum_{i=0}^{N-1}\zeta^{i^2}.    
\end{equation}
\end{definition} Note that this is a general definition of the braid element, which goes beyond the diagrams above, since we allow for $|k-l|\neq1$, which includes the local (nearest-neighbor) braid operators as a special case. We hasten to add that the terminology ``braid element'' at this point is only suggestive. To justify this terminology one has to prove that the braid elements satisfy braiding relations, in particular the Yang-Baxter equation, which is the subject of the section titled Applications on the Golden Rule.

\begin{remark}
	$\omega$ has modulus 1 (this fact is proven in Proposition 2.15 in \cite{Jaffe2017}), implying that \begin{equation}
	    b_{kl}^{\dagger} = b_{lk}
	\end{equation} for $k\neq l$.
\end{remark}

Thus, in terms of terminology, we will refer to the positive braids as just braids, and the negative braids as adjoint braids.

\subsection{Graphical Representation of the Representation-Theoretic Axioms}

In previous work \cite{Lin1}, two axioms were presented as a way to abstract certain high-level properties of the generalized Clifford algebras. It was shown that these 2 axioms are satisfied by an explicit construction. 
 These axioms will now be converted into graphical form.  

\textbf{Axiom 1}:
Let $\mathcal{V}^{N^n}(\mathbb{C})$ be a complex vector space upon which the generalized Clifford algebra is realized as unitary $N^n$ by $N^n$ matrix operators. Assume that there exists a state (which we call the ground state) which is a tensor of states $\ket{\Omega}$, $\ket{\Omega}^{\otimes n}$, that satisfies the following algebraic identity:
$$ c_{2k-1} \ket{\Omega}^{\otimes n} = \zeta\, c_{2k} \ket{\Omega}^{\otimes n}$$
for all $k=1,2,\ldots, n$, where $\zeta$ is a square root of $q$ such that $\zeta^{N^2}=1$. 

In addition, for each qudit, the projector $E_{k}$ onto the $k$th qudit's ground state $\ket{\Omega}$ is assumed to satisfy 
$$ c_{2k-1} E_{k} = \zeta\, c_{2k} E_{k}.$$

\textbf{Axiom 2: Scalar product}:
The set $\{c_2^{a_1} c_4^{a_2} \ldots c_{2n}^{a_n} \ket{\Omega}^{\otimes n}: a_i=0,1, \ldots, N-1\}$ is an orthonormal basis for $\mathcal{V}^{N^n}(\mathbb{C})$.

These axioms are now shown to give rise to basic graphical identities. The algebraic identities
$$
c_{i} c_{j} = q c_{j} c_{i}
$$
for $i<j$,
$$
c_{i}^N=1
$$
for all $i=1,2,\ldots, 2n$,
as well as
$$
c_{2k-1} E_k=\zeta c_{2k} E_k
$$
tell us that

\begin{equation}
\raisebox{-.3cm}{	\tikz{
		\draw (1\mn,0)--(1\mn,1);
		\draw (0\mn,0)--(0\mn,1);
		\draw (1\mn,-1.1)--(1\mn,-0.1);
		\draw (0\mn,-1.1)--(0\mn,-0.1);
			\node at (0.5\mn,-2/4) {\size{$1$}};
			\node at (-1\mn,-2.5/4) {\size{$.$}};
			\node at (-1\mn,2/4) {\size{$.$}};
			\node at (-1.5\mn,-2.5/4) {\size{$.$}};
			\node at (-1.5\mn,2/4) {\size{$.$}};
}} \; \! \!
\raisebox{-.3cm}{	\tikz{
		\draw (1\mn,0)--(1\mn,1);
		\draw (0\mn,0)--(0\mn,1);
		\draw (1\mn,-1.1)--(1\mn,-0.1);
		\draw (0\mn,-1.1)--(0\mn,-0.1);
		\node at (-1\mn,-2.5/4) {\size{$.$}};
		\node at (-1\mn,2/4) {\size{$.$}};
		\node at (-1.5\mn,-2.5/4) {\size{$.$}};
		\node at (-1.5\mn,2/4) {\size{$.$}};
	
}} \; \! \!
\raisebox{-.3cm}{	\tikz{
		\draw (1\mn,0)--(1\mn,1);
		\draw (0\mn,0)--(0\mn,1);
		\node at (0.5\mn,3/4) {\size{$1$}};
			\draw (1\mn,-1.1)--(1\mn,-0.1);
		\draw (0\mn,-1.1)--(0\mn,-0.1);
}} \;
\raisebox{.5cm}{	\tikz{
		\node at (0.5\mn,4/4) {\size{$=$}};
		}} \; \! \! \! 
	\raisebox{.5cm}{	\tikz{
			\node at (0.5\mn,4/4) {\size{$q$}};
	}} \!
	\raisebox{-.3cm}{	\tikz{
			\draw (1\mn,0)--(1\mn,1);
			\draw (0\mn,0)--(0\mn,1);
			\draw (1\mn,-1.1)--(1\mn,-0.1);
			\draw (0\mn,-1.1)--(0\mn,-0.1);
			\node at (0.5\mn,3/4) {\size{$1$}};
					\node at (-1\mn,-2.5/4) {\size{$.$}};
			\node at (-1\mn,2/4) {\size{$.$}};
			\node at (-1.5\mn,-2.5/4) {\size{$.$}};
			\node at (-1.5\mn,2/4) {\size{$.$}};
	}} \; \! \!
	\raisebox{-.3cm}{	\tikz{
			\draw (1\mn,0)--(1\mn,1);
			\draw (0\mn,0)--(0\mn,1);
			\draw (1\mn,-1.1)--(1\mn,-0.1);
			\draw (0\mn,-1.1)--(0\mn,-0.1);
				\node at (-1\mn,-2.5/4) {\size{$.$}};
			\node at (-1\mn,2/4) {\size{$.$}};
			\node at (-1.5\mn,-2.5/4) {\size{$.$}};
			\node at (-1.5\mn,2/4) {\size{$.$}};
	}} 
	\raisebox{-.3cm}{	\tikz{
			\draw (1\mn,0)--(1\mn,1);
			\draw (0\mn,0)--(0\mn,1);
				\node at (0.5\mn,-2/4) {\size{$1$}};
			\draw (1\mn,-1.1)--(1\mn,-0.1);
			\draw (0\mn,-1.1)--(0\mn,-0.1);
	}} \;
\end{equation}
i.e. when the primitive for $c_j$ precedes that for $c_i$, swapping the order of primitives yields a factor of $q$, for $i<j$, and also that

\begin{equation}\raisebox{-.3cm}{	\tikz{
		\draw (1\mn,0)--(1\mn,1);
		\draw (0\mn,0)--(0\mn,1);
		\node at (-1\mn,2/4) {\size{$.$}};
		\node at (-1.5\mn,2/4) {\size{$.$}};
}}
\raisebox{-.3cm}{	\tikz{
		\draw (1\mn,0)--(1\mn,1);
		\draw (0\mn,0)--(0\mn,1);
		\node at (1.5\mn,3/4) {\size{$\scalebox{0.75}{\textit{N}}$}};
		\node at (-1\mn,2/4) {\size{$.$}};
		\node at (-1.5\mn,2/4) {\size{$.$}};
}}
\raisebox{-.3cm}{	\tikz{
		\draw (1\mn,0)--(1\mn,1);
		\draw (0\mn,0)--(0\mn,1);
}}\;
=\raisebox{-.3cm}{	\tikz{
		\draw (1\mn,0)--(1\mn,1);
		\draw (0\mn,0)--(0\mn,1);
		\node at (-1\mn,2/4) {\size{$.$}};
		\node at (-1.5\mn,2/4) {\size{$.$}};
}}
\raisebox{-.3cm}{	\tikz{
		\draw (1\mn,0)--(1\mn,1);
		\draw (0\mn,0)--(0\mn,1);
		\node at (0.5\mn,1/4) {\size{$\scalebox{0.75}{\textit{N}}$}};
		\node at (-1\mn,2/4) {\size{$.$}};
		\node at (-1.5\mn,2/4) {\size{$.$}};
}}
\raisebox{-.3cm}{	\tikz{
		\draw (1\mn,0)--(1\mn,1);
		\draw (0\mn,0)--(0\mn,1);
}}\;=\raisebox{-.3cm}{	\tikz{
		\draw (1\mn,0)--(1\mn,1);
		\draw (0\mn,0)--(0\mn,1);
		\node at (-1\mn,2/4) {\size{$.$}};
		\node at (-1.5\mn,2/4) {\size{$.$}};
}}
\raisebox{-.3cm}{	\tikz{
		\draw (1\mn,0)--(1\mn,1);
		\draw (0\mn,0)--(0\mn,1);
		\node at (-1\mn,2/4) {\size{$.$}};
		\node at (-1.5\mn,2/4) {\size{$.$}};
}}
\raisebox{-.3cm}{	\tikz{
		\draw (1\mn,0)--(1\mn,1);
		\draw (0\mn,0)--(0\mn,1);
}}\end{equation}
and
\begin{equation}\raisebox{-.3cm}{	\tikz{
		\draw (1\mn,0)--(1\mn,1);
		\draw (0\mn,0)--(0\mn,1);
		\draw (1\mn,-1.1)--(1\mn,-0.1);
		\draw (0\mn,-1.1)--(0\mn,-0.1);
		\node at (-1\mn,-2.5/4) {\size{$.$}};
		\node at (-1\mn,2/4) {\size{$.$}};
		\node at (-1.5\mn,-2.5/4) {\size{$.$}};
		\node at (-1.5\mn,2/4) {\size{$.$}};
}} \; \! \!
\raisebox{-.3cm}{	\tikz{
		\fqudit {-2}{0\nn}{0.5\mn}{.75\nn}{\phantom{ll}}
		\fmeasure {-2}{3\nn}{0.5\mn}{.75\nn}{}
		\draw (6\mn,-1.1)--(6\mn,-0.1);
		\draw (5\mn,-1.1)--(5\mn,-0.1);
		\node at (6.5\mn,-3/4) {\size{$1$}};
		\node at (4\mn,-2.5/4) {\size{$.$}};
		\node at (4\mn,2/4) {\size{$.$}};
		\node at (3.5\mn,-2.5/4) {\size{$.$}};
		\node at (3.5\mn,2/4) {\size{$.$}};
}} \; \! \!
\raisebox{-.3cm}{	\tikz{
		\draw (1\mn,0)--(1\mn,1);
		\draw (0\mn,0)--(0\mn,1);
		\draw (1\mn,-1.1)--(1\mn,-0.1);
		\draw (0\mn,-1.1)--(0\mn,-0.1);
}} \;
\raisebox{.5cm}{	\tikz{
		\node at (0.5\mn,4/4) {\size{$=$}};
}} \; \! \! \! 
\raisebox{.5cm}{	\tikz{
		\node at (0.5\mn,4/4) {\size{$\zeta$}};
}} \!
\raisebox{-.3cm}{	\tikz{
		\draw (1\mn,0)--(1\mn,1);
		\draw (0\mn,0)--(0\mn,1);
		\draw (1\mn,-1.1)--(1\mn,-0.1);
		\draw (0\mn,-1.1)--(0\mn,-0.1);
		\node at (-1\mn,-2.5/4) {\size{$.$}};
		\node at (-1\mn,2/4) {\size{$.$}};
		\node at (-1.5\mn,-2.5/4) {\size{$.$}};
		\node at (-1.5\mn,2/4) {\size{$.$}};
}} \; \! \!
\raisebox{-.3cm}{	\tikz{
		\fqudit {-2}{0\nn}{0.5\mn}{.75\nn}{\phantom{ll}}
		\fmeasure {-2}{3\nn}{0.5\mn}{.75\nn}{}
		\draw (6\mn,-1.1)--(6\mn,-0.1);
		\draw (5\mn,-1.1)--(5\mn,-0.1);
		\node at (5.5\mn,-3/4) {\size{$1$}};
		\node at (4\mn,-2.5/4) {\size{$.$}};
		\node at (4\mn,2/4) {\size{$.$}};
		\node at (3.5\mn,-2.5/4) {\size{$.$}};
		\node at (3.5\mn,2/4) {\size{$.$}};
}} 
\raisebox{-.3cm}{	\tikz{
		\draw (1\mn,0)--(1\mn,1);
		\draw (0\mn,0)--(0\mn,1);
		\draw (1\mn,-1.1)--(1\mn,-0.1);
		\draw (0\mn,-1.1)--(0\mn,-0.1);
}} \;.
\end{equation}

Furthermore, the vector identity 
$$
c_{2k-1} \ket{\Omega}^{\otimes n} = \zeta c_{2k} \ket{\Omega}^{\otimes n}
$$
yields the diagrammatic ``identity''

\begin{equation}
\raisebox{-.3cm}{	\tikz{
		\fqudit {-2}{0\nn}{0.5\mn}{.75\nn}{\phantom{ll}}
		\draw (6\mn,-1.1)--(6\mn,-0.1);
		\draw (5\mn,-1.1)--(5\mn,-0.1);
		\node at (4\mn,-2.5/4) {\size{$.$}};
		\node at (4\mn,1/4) {\size{$.$}};
		\node at (3.5\mn,-2.5/4) {\size{$.$}};
		\node at (3.5\mn,1/4) {\size{$.$}};
}} \; \! \!
\raisebox{-.3cm}{	\tikz{
		\fqudit {-2}{0\nn}{0.5\mn}{.75\nn}{\phantom{ll}}
		\draw (6\mn,-1.1)--(6\mn,-0.1);
		\draw (5\mn,-1.1)--(5\mn,-0.1);
		\node at (6.5\mn,-3/4) {\size{$1$}};
		\node at (4\mn,-2.5/4) {\size{$.$}};
		\node at (4\mn,1/4) {\size{$.$}};
		\node at (3.5\mn,-2.5/4) {\size{$.$}};
		\node at (3.5\mn,1/4) {\size{$.$}};
}} 
\raisebox{-.3cm}{	\tikz{
		\fqudit {-2}{0\nn}{0.5\mn}{.75\nn}{\phantom{ll}}
		\draw (6\mn,-1.1)--(6\mn,-0.1);
		\draw (5\mn,-1.1)--(5\mn,-0.1);
}} \;
\raisebox{.15cm}{	\tikz{
		\node at (0.5\mn,1/4) {\size{$=$}};
}} \; \! \! \! 
\raisebox{.1cm}{	\tikz{
		\node at (0.5\mn,1/4) {\size{$\zeta$}};
}} \!
\raisebox{-.3cm}{	\tikz{
		\fqudit {-2}{0\nn}{0.5\mn}{.75\nn}{\phantom{ll}}
\draw (6\mn,-1.1)--(6\mn,-0.1);
\draw (5\mn,-1.1)--(5\mn,-0.1);
\node at (4\mn,-2.5/4) {\size{$.$}};
\node at (4\mn,1/4) {\size{$.$}};
\node at (3.5\mn,-2.5/4) {\size{$.$}};
\node at (3.5\mn,1/4) {\size{$.$}};
}} \; \! \!
\raisebox{-.3cm}{	\tikz{
		\fqudit {-2}{0\nn}{0.5\mn}{.75\nn}{\phantom{ll}}
		\draw (6\mn,-1.1)--(6\mn,-0.1);
		\draw (5\mn,-1.1)--(5\mn,-0.1);
		\node at (5.5\mn,-3/4) {\size{$1$}};
		\node at (4\mn,-2.5/4) {\size{$.$}};
		\node at (4\mn,1/4) {\size{$.$}};
		\node at (3.5\mn,-2.5/4) {\size{$.$}};
		\node at (3.5\mn,1/4) {\size{$.$}};
}} 
\raisebox{-.3cm}{	\tikz{
		\fqudit {-2}{0\nn}{0.5\mn}{.75\nn}{\phantom{ll}}
		\draw (6\mn,-1.1)--(6\mn,-0.1);
		\draw (5\mn,-1.1)--(5\mn,-0.1);
}} \;.
\end{equation}

An additional identity which is useful \cite{Jaffe2017} is the following:
\begin{lemma}
	\label{ABlemma}
	\begin{equation}
	c_i^a c_j^b = q^{ab} c_j^b c_i^a
	\end{equation}
	for $i<j$, $a$, $b$ integers.
\end{lemma}
\begin{proof}
	By double induction on $a$ and $b$.
\end{proof}

Another identity, due to \cite{Jaffe2017}, is
\begin{lemma}
	\label{caplemma}
	\begin{equation}
	c_{2i-1}^a E_i = \zeta^{a^2} c_{2i}^a E_{2i}
	\end{equation}
	for $i=1,2,\ldots,n$, $a$ an integer.
\end{lemma}
\begin{proof}
	By induction.
\end{proof}

\section{Graphical Calculus at the Level of the Multi-Qudit Operators}
\label{alg}

Our aim in this section is to obtain a large swath of identities, which are related to the graphical representation we have presented, but for which we provide purely algebraic proofs. At the heart of the results of this section are a new ``charge-braid'' identity that answers an open question due to Jaffe, namely, how to bring the charge ``over'' the braid when $N\neq 2$. This seemingly innocuous result is used to great effect, by using the structural property that the generalized Clifford algebra generated by $c_1, c_2, \ldots, c_{2n}$ has trivial center. In particular, we provide an algebraic proof, using the proof strategy based on this structural characterization, that the braid elements $b_{kl}$ satisfy many Yang-Baxter equations. Furthermore, we construct a general solution to the braid group relations, which enables us to resolve an open question of \cite{Cobanera2014} for the case where $N$ is even.

\subsection{Structural Properties of the Generalized Clifford Algebras}

\begin{proposition}
	\label{basisprop}
	The set $\{ c_1^{r_1} c_2^{r_2} \cdots c_{2n}^{r_{2n}}: r_1, r_2, \ldots r_{2n} = 0,1,\ldots N-1\}$ is a basis for the generalized Clifford algebra $\mathcal{C}_{2n}^{(N)}$.
\end{proposition}
\begin{proof}
	Any element of the generalized Clifford algebra is a finite sum of elements of the form $\alpha \, c_{k_1}^{\epsilon_1} c_{k_2}^{\epsilon_2} \cdots c_{k_m}^{\epsilon_m}$ for $\alpha \in \mathbb{C}$, $m$ a positive integer, $k_i$ in the index set $I_{2n} = \{1,2,\cdots, 2n\}$, and $\epsilon_i \in \{1,-1\}$ for $i=1,2,\ldots,m$. By repeatedly applying the relations $c_{k_i}^{-1} = c_{k_i}^{N-1}$ and $c_i c_j = q c_j c_i$ for $i<j$ to swap the order of multiplication, we can put each term in the sum into \textbf{normal form}, by which we mean that the term is of the form $\beta_{r_1 r_2 \ldots r_{2n}} \, c_{1}^{r_1} c_{2}^{r_2} \cdots c_{2n}^{r_{2n}}$, for $r_i \in \{0,1,2,\ldots, N-1\}$. Thus, we obtain that every element $x$ of the generalized Clifford algebra is prescribed by a sum given by $$x = \sum_{r_1, r_2, \ldots r_{2n} = 0,1,\ldots N-1} x_{r_1 r_2 \ldots r_{2n}} c_1^{r_1} c_2^{r_2} \cdots c_{2n}^{r_{2n}}. $$
	
	Now we want to show that $x=0$ in the algebra if and only if $x_{r_1 r_2 \cdots r_{2n}}=0$ for all indices, i.e. the set $\{ c_1^{r_1} c_2^{r_2} \cdots c_{2n}^{r_{2n}}: r_1, r_2, \ldots r_{2n} = 0,1,\ldots N-1\}$ is a basis. The if direction is obviously true. For the only if direction, suppose $x=0$. Then multiplying $x$ by any product of generators $c_i$ also yields zero. It is clear that we can multiply $x$ on the left by the product $c_{2n}^{-r_{2n}} c_{2n-1}^{-r_{2n-1}}\cdots c_{2}^{-r_2} c_{1}^{-r_1}$ so that the constant term of $c_{2n}^{-r_{2n}} c_{2n-1}^{-r_{2n-1}}\cdots c_{2}^{-r_2} c_{1}^{-r_1} x$ is $x_{r_1 r_2 \cdots r_{2n}}$. Thus, without loss of generality, it suffices to show that if $x=0$, then its constant term must vanish. Then the rest of the coefficients all vanish by applying the same result to \\ \noindent $c_{2n}^{-r_{2n}} c_{2n-1}^{-r_{2n-1}}\cdots c_{2}^{-r_2} c_{1}^{-r_1} x$ for each index tuple.
	
	To show that the constant term must vanish, we use an operator method. Consider the set of operators $L_k(y) = \sum_{i=0}^{N-1} c_k^{i} y c_k^{-i}$, and let $L_k^{(l)} := L_k^{(l-1)} \circ L_k$  and $L_k^{(0)}:=1$ define $L_k^{(l)}$ iteratively. Then the operator $M_k = \sum_{l=0}^{N-1} L_k^{(l)}$ acting on a term $c_1^{r_1} c_2^{r_2} \cdots c_{2n}^{r_{2n}}$ yields 
 \begin{equation}
 \left(\sum_{l=0}^{N-1}(q^{-\sum_{i<k}r_i +\sum_{i>k} r_i})^l\right) c_1^{r_1} c_2^{r_2} \cdots c_{2n}^{r_{2n}} = N \delta(\sum_{i<k} r_i, \sum_{i>k} r_i) c_1^{r_1} c_2^{r_2} \cdots c_{2n}^{r_{2n}},  
 \end{equation} where $\delta(a,b) := 1$ if $a\equiv b \text{ mod } N$, and $0$ otherwise. Acting on $x$ by the commuting operators $\frac{1}{N} M_k$ (which all have a diagonal action on $c_1^{r_1} c_2^{r_2} \cdots c_{2n}^{r_{2n}}$) thus projects $x$ down to 
	\begin{equation}
		(\prod_{k=1}^{2n} \frac{1}{N} M_k )(x) =  \sum_{r_1, r_2, \ldots r_{2n} = 0,1,\ldots N-1} \left(\prod_{k=1}^{2n}\delta(\sum_{i<k} r_i, \sum_{i>k} r_i) \right) x_{r_1 r_2 \ldots r_{2n}} c_1^{r_1} c_2^{r_2} \cdots c_{2n}^{r_{2n}}.
	\end{equation}

	We first claim that the only terms that survive are those for which $r_{k}+r_{k+1} = 0 \text{ mod } N$ for $k=1,2,\ldots, 2n-1$. This can be seen since 
	\begin{equation}
		\sum_{i<k} r_i = \sum_{i>k} r_i \Rightarrow
		2\sum_{i<k}r_i + r_k = \sum_{i=1}^{2n} r_i
	\end{equation}
	for all $k=1,2,\ldots, 2n$ implies that 
	\begin{equation}
		2\sum_{i<k}r_i + r_k = 2 \sum_{i<k+1} r_i + r_{k+1} = 2 \sum_{i<k} r_i + 2 r_k + r_{k+1}
	\end{equation}
for all $k=1,2,\ldots, 2n-1$,
 	and so
 	\begin{equation}
 		r_{k} + r_{k+1} = 0 \text{ mod } N,
 	\end{equation}
	 as desired.
	 
	 As a result, we further obtain that 
	 \begin{equation*}
	 	r_{2n} = 0 
	 \end{equation*}
 	since 
 	\begin{equation*}
 		\sum_{i<2n-1} r_i = (r_1+r_2) + (r_3+r_4)+ \cdots + (r_{2n-3} + r_{2n-2}) = 0 = r_{2n}.
 	\end{equation*}
 	Finally, using $r_{k}+r_{k+1} =0$ for $k=1,2,\ldots, 2n-1$ we obtain that $r_{k}=0$ for all $k=1,2,\ldots, 2n$. Hence the constant term is the only term left, and must equal $0$ since  $M_k(0) = 0$.

\end{proof}

\begin{proposition}[Golden Rule]
	\label{GoldenRule}
	The generalized Clifford algebra $\mathcal{C}_{2n}^{(N)}$ has trivial center, i.e. the only elements that commute with all elements of the generalized Clifford algebra are $\mathbb{C}1$.
\end{proposition}

\begin{proof}

Every element of the generalized Clifford algebra is prescribed by a sum given by $$x = \sum_{r_1, r_2, \ldots r_{2n} = 0,1,\ldots N-1} x_{r_1 r_2 \ldots r_{2n}} c_1^{r_1} c_2^{r_2} \cdots c_{2n}^{r_{2n}}. $$

Using the basis property (Proposition \ref{basisprop}), it becomes simple to show that the algebra has trivial center. Note that the basis property implies uniqueness of the sum decomposition. Let $x$ lie in the center of the algebra, and $x\neq 0$. Then there is an index label $r_1,r_2, \cdots, r_{2n}$ such that $x_{r_1 r_2 \cdots r_{2n}}\neq 0$. Note that $x c_1 = c_1 x$ implies that $x_{r_1 r_2 \cdots r_{2n}} = q^{-(r_2 + r_3 +\cdots r_{2n})} x_{r_1 r_2 \cdots r_{2n}}$ by comparing the coefficient of $c_1^{r_1+1} c_2^{r_2} \cdots c_{2n}^{r_{2n}}$. Thus, $r_2 + r_3 + \cdots + r_{2n}=0$. Similarly, $x c_k = c_k x $ implies that $q^{-\sum_{i<k} r_i} x_{r_1 r_2 \cdots r_{2n}} q^{\sum_{i>k} r_i}x_{r_1 r_2 \cdots r_{2n}}=1$ and so 
\begin{equation}
\sum_{i=1}^{2n}\epsilon_{ik} r_i= 0 \text{ (mod }N),
\end{equation} for $k$ from $1$ to $2n$, where $\epsilon_{ik}=1$ if $i<k$ and $-1$ if $i>k$ and $0$ if $i=k$, yielding $2n$ equations in $2n$ unknowns. Equivalently,
\begin{equation}
    \sum_{i<k} r_i = \sum_{i>k} r_i \text{ (mod } N)
\end{equation}
for all $k=1,2,\cdot, 2n$. Since in Proposition \ref{basisprop}, it was shown that this set of equations is uniquely solved by $r_1 = r_2 = \cdots = r_{2n} = 0$, it follows that $x$ is a multiple of the identity $1$. 
\end{proof}

\subsection{An ``Intertwining'' Approach for New Identities for the Generalized Clifford  Algebra}

\subsubsection{A Systematic Procedure}

The golden rule of Proposition \ref{GoldenRule} allows us to give a systematic procedure for proving identities in the algebra. The basis of the procedure is the following proposition:
\begin{proposition}
	\label{identitycertificate1}
	Let $x$, $y$ lie in the generalized Clifford algebra, and suppose $y$ is invertible. Further assume that the constant terms of $x$ and $y$ are nonzero. Then $x=y$ if and only if $y^{-1} x$ lies in the center of the generalized Clifford algebra, and the constant term in $x$ agrees with the constant term in $y$.
\end{proposition}
\begin{proof}
	Clearly, the only if direction is true since $x=y$ implies $y^{-1} x = 1$. 	For the if direction, if $y^{-1} x$ lies in the center, by the golden rule, $y^{-1} x \in \mathbb{C}1$, i.e. $y = \alpha x$. In the proof of proposition \ref{GoldenRule}, we showed that this implies that all terms of $y$ and $\alpha x$ agree, in particular the constant terms. By hypothesis, the constant terms of $y$ and $x$ agree and are nonzero, so $\alpha = 1$.
\end{proof}

We now provide a concrete way to show that an element lies in the center of the generalized Clifford algebra.

\begin{proposition}
	\label{identitycertificate2}
	An element $x$ lies in the center of the generalized Clifford algebra if and only if it commutes with $c_i$ for each $i=1,2,\ldots, 2n$.
\end{proposition}
\begin{proof}
	The only if direction is clearly true.
	
	For the if direction, any element $y$ in the algebra has a unique decomposition as $$y = \sum_{r_1, r_2, \ldots r_{2n} = 0,1,\ldots N-1} y_{r_1 r_2 \ldots r_{2n}} c_1^{r_1} c_2^{r_2} \cdots c_{2n}^{r_{2n}}. $$
	
	By iterative commutation, using the commutation property of $x$ with $c_i$, one can show that 
	$x \,	c_1^{r_1} c_2^{r_2} \cdots c_{2n}^{r_{2n}} = c_1^{r_1} c_2^{r_2} \cdots c_{2n}^{r_{2n}} \, x$. Multiplying by the constant prefactor and summing over the indices, one obtains that $x y = y x$, as desired, for arbitrary $y$ in the algebra.
\end{proof}

\subsubsection{Intertwining Identities}

By intertwining identities, we mean identities of the form $bx = yb$. In this section, we present the following new intertwining identity for the braid $b_{kl}$. We first give a direct proof, and then give an alternate proof which involves some intermediate intertwining identities, the particular concatenation of which may have more general applications. This identity significantly generalizes a theorem of Jaffe and Liu \cite{Jaffe2017} (Theorem 8.2), which is the special case for $a=0$. 
\begin{theorem}
	\label{braideqn}
    \begin{equation}
	b_{kl} c_k^a c_l^b = q^{a^2+ab} c_k^{2a+b} c_l^{-a} b_{kl}
    \end{equation}
	for $k<l$.
\end{theorem}
\begin{proof}
	Since $b_{kl} = \frac{\omega^{1/2}}{\sqrt{N}} \sum_{i=0}^{N-1} c_k^i c_l^{-i}$, it suffices to show that $$\left(\sum_{i=0}^{N-1} c_k^i c_l^{-i}\right)c_k^a c_l^b = q^{a^2+ab} c_k^{2a+b} c_l^{-a} \left(\sum_{i=0}^{N-1} c_k^i c_l^{-i}\right).$$
	
	Applying lemma \ref{ABlemma}, the LHS becomes 
	\begin{equation}
	\sum_{i=0}^{N-1} q^{ai} c_k^{a+i} c_{l}^{b-i}
	\end{equation}
	and the RHS becomes
	\begin{equation}
	\sum_{i=0}^{N-1} q^{a^2+ab} q^{ai} c_k^{2a+b+i} c_l^{-a-i}.
	\end{equation}
	
	By shifting the index of summation from $i$ to $i+a+b$ in the LHS, the LHS becomes 
	\begin{equation}
	\sum_{i=0}^{N-1} q^{a(i+a+b)} c_k^{2a+b+i} c_{l}^{-a-i}
	\end{equation}
	which is just the RHS.	
\end{proof}

In terms of the graphical calculus, we economically write down the following diagrammatic identity, which is specific to $b_{12}$ and the generalized Clifford algebra with only 2 generators $c_1$, $c_2$:
\begin{equation}
\!	\raisebox{-.65cm}{\scalebox{0.8}{
		\tikz{
			\fbraid{0\mn}{0}{3\mn}{3\nn}
			\node at (4\mn,4.0\nn) {\size{$a$}};
	
			\node at (0.9\mn,5.2\nn) {\size{$b$}};
			\draw (3\mn,2)--(3\mn,1);
			\draw (0\mn,2)--(0\mn,1);	
}}} \;
=q^{a^2+ab}
\raisebox{-0.65cm}{\scalebox{.8}{
\tikz{
\fbraid{0\mn}{0}{3\mn}{3\nn}
\node at (4.75\mn,-2.0\nn) {\size{$2a+b$}};
\node at (1.0\mn,-1.0\nn) {\size{$-a$}};
\draw (3\mn,-1)--(3\mn,0);
\draw (0\mn,-1)--(0\mn,0);	
}}} \;
\end{equation}

It is convenient to also write down the corresponding identity for the adjoint braid:

\begin{corollary}
	\label{adjointbraideqn}
	\begin{equation}
    b_{lk} c_k^r c_l^s = q^{rs+s^2} c_k^{-s} c_l^{r+2s} b_{lk}.
    \end{equation}
	for $k<l$, and $r$,$s$ integers.
\end{corollary}
\begin{proof}
	The adjoint of the identity in \ref{braideqn} is $c_l^{-b} c_k^{-a} b_{lk} = q^{-a^2-ab} b_{lk} c_l^a c_k^{-2a-b}$, which becomes $q^{-ab} c_k^{-a} c_l^{-b} b_{lk}=q^{a^2} b_{lk} c_k^{-2a-b} c_l^a$ upon commutation. Now we let $r = -2a-b$, $s=a$, so 
    \begin{equation}
    b_{lk} c_k^r c_l^s = q^{rs+s^2} c_k^{-s} c_l^{r+2s} b_{lk},
    \end{equation}
    which gives the desired result.
\end{proof}
The corresponding diagrammatic identity for the adjoint braid $b_{21}$ arising from Corollary \ref{adjointbraideqn} for the generalized Clifford algebra with two generators $c_1$, $c_2$ is 
\begin{equation}
\!	\raisebox{-.65cm}{\scalebox{0.8}{
		\tikz{
			\fbraid{3\mn}{0}{0\mn}{3\nn}
			\node at (4\mn,4.0\nn) {\size{$r$}};
			
			\node at (0.9\mn,5.2\nn) {\size{$s$}};
			\draw (3\mn,2)--(3\mn,1);
			\draw (0\mn,2)--(0\mn,1);	
}}} \;
=q^{rs+s^2}
\raisebox{-0.65cm}{\scalebox{.8}{
		\tikz{
			\fbraid{3\mn}{0}{0\mn}{3\nn}
			\node at (4.00\mn,-2.0\nn) {\size{$-s$}};
			\node at (1.5\mn,-1.0\nn) {\size{$r+2s$}};
			\draw (3\mn,-1)--(3\mn,0);
			\draw (0\mn,-1)--(0\mn,0);	
}}} \;
\end{equation}

We now pursue an alternate route to proving Equation \ref{braideqn}, which illuminates complementary aspects. We start with an intertwining identity which is a commutation relation:
\begin{lemma}
	\label{commuting}
	\begin{equation}
    (c_k^bc_l^{-b})(c_k^ac_l^{-a})=(c_k^ac_l^{-a})(c_k^bc_l^{-b}) 
    \end{equation}
	for $k<l$.
\end{lemma}
\begin{proof}
	Applying lemma \ref{ABlemma} to LHS yields $q^{ab} c_k^{a+b} c_l^{-(a+b)}$; applying lemma \ref{ABlemma} to RHS yields $q^{ab} c_k^{a+b} c_l^{-(a+b)}$. Thus, LHS=RHS.
\end{proof}
We also note that the following commutation relation holds as well:
\begin{lemma}
	\label{outercommuting}
	\begin{equation}
    (c_k^ac_l^{-a}) c_p = c_p (c_k^ac_l^{-a})
    \end{equation}
	for $k<l$ and $p$ satisfies $p< k<l$ or $p>l>k$.
\end{lemma}
\begin{proof}
	If $k<l<p$, commuting $c_p$ past (in front of) $c_l^{-a}$ in the LHS yields $q^{-a}$; commuting it past $c_k^a$ then yields an additional factor $q^a$. So we obtain the RHS. A similar proof applies for the case $p<k<l$.
\end{proof}
Now comes the exciting part. Since the braid $b_{kl}$ is a sum of elements of the form $c_k^{i} c_l^{-i}$, it follows by linearity that
\begin{lemma}
	\label{braidcommuting}
	\begin{equation}
        b_{kl} \, c_k^{a} c_l^{-a} = c_k^{a} c_l^{-a} \, b_{kl}
    \end{equation}
	for $k<l$.
\end{lemma}
\begin{proof}
	By linear extension of Lemma \ref{commuting}.
\end{proof}
Now we use a simple result due to Jaffe and Liu \cite{Jaffe2017} (Theorem 8.2) :
\begin{lemma}
	\label{chargecommuting}
	\begin{equation}
        b_{kl} c_l = c_k b_{kl}
 	\end{equation}
  for $k<l$.
\end{lemma}
\begin{proof}
	It suffices to show that 
 \begin{equation}
 \left(\sum_{i=0}^{N-1} c_k^i c_l^{-i}\right) c_l = c_k \left(\sum_{i=0}^{N-1} c_k^i c_l^{-i}\right).
 \end{equation}
Collecting terms, it is equivalent to show that 
\begin{equation}
\sum_{i=0}^{N-1} c_k^i c_l^{-(i-1)}= \sum_{i=0}^{N-1} c_k^{i+1} c_l^{-i}.
\end{equation}
It is clear that the two are equal since the RHS is just the LHS with $i$ shifted to $i-1$.
\end{proof}

It remains but to combine lemmas \ref{braidcommuting} and \ref{chargecommuting}, giving us an alternate proof of Theorem \ref{braideqn}:
\begin{proof}[Alternate Proof of Theorem \ref{braideqn}]
	We want to show that \begin{equation}b_{kl} c_k^a c_l^b = q^{a^2+ab} c_k^{2a+b} c_l^{-a} b_{kl}\end{equation}
	for $k<l$. To use lemmas \ref{braidcommuting} and \ref{chargecommuting}, we rewrite $b_{kl} c_k^a c_l^b$ as $b_{kl} c_k^{a} c_l^{-a} c_l^{a+b}$. This becomes $c_k^{a} c_l^{-a} b_{kl} c_l^{a+b}$ after commuting past the braid, and then $c_k^a c_l^{-a} c_k^{a+b} b_{kl}$ after applying lemma \ref{chargecommuting} $a+b$ times. Finally, applying lemma \ref{ABlemma} to the middle two terms yields $q^{a^2+ab} c_k^{2a+b} c_l^{-a}b_{kl}$ as desired.
\end{proof}

\subsubsection{A Grading of the Generalized Clifford Algebra}
We now interpret the previous section's intertwining identities in terms of a grading on the generalized Clifford algebra. In particular, it is observed that the new charge-braid identity in Proposition \ref{braideqn} is a consequence of a particular property of neutral pairings of $c_k$ and $c_l$. First, we define a charge operator $C$:
\begin{definition}
	Define $C$ by linear extension of its action on the basis:
	\begin{equation}
        C (c_1^{r_1} c_2^{r_2} \cdots c_{2n}^{r_{2n}}) := q^{r_1+r_2+\cdots+r_{2n}} c_1^{r_1} c_2^{r_2} \cdots c_{2n}^{r_{2n}}	    
	\end{equation}
	for all integer indices $r_i$.	We call $r_1+r_2+\cdots + r_{2n}$ the \textbf{charge} of the basis element, following \cite{LiuPNAS}, which is well-defined modulo $N$. This terminology of an element's charge is also applicable for linear combinations of basis elements with the same charge.
\end{definition}

Then, lemma \ref{commuting} tells us that eigenstates of $C$ of eigenvalue 1 which lie in the subalgebra generated by $c_k$, $c_l$ commute. We call eigenstates of $C$ with eigenvalue 1 \textit{neutral}.

Graphically, we can describe this commutation relation \ref{commuting} for the algebra generated by $c_1$ and $c_2$ as
\begin{equation}\raisebox{-1cm}{	\tikz{
		\draw (3\mn,0)--(3\mn,1);
		\draw (0\mn,0)--(0\mn,1);
		\draw (3\mn,-1.1)--(3\mn,-0.1);
		\draw (0\mn,-1.1)--(0\mn,-0.1);
		\node at (1.0\mn,-1.5/4) {\size{$-a$}};
		\node at (1.0\mn,3/4) {\size{$-b$}};
		\node at (3.75\mn,1/4) {\size{$b$}};
		\node at (3.75\mn,-3.5/4) {\size{$a$}};
}} \; 
=
\raisebox{-1cm}{	\tikz{
		\draw (3\mn,0)--(3\mn,1);
		\draw (0\mn,0)--(0\mn,1);
		\draw (3\mn,-1.1)--(3\mn,-0.1);
		\draw (0\mn,-1.1)--(0\mn,-0.1);
		\node at (1.0\mn,-1.5/4) {\size{$-b$}};
		\node at (1.0\mn,3/4) {\size{$-a$}};
		\node at (3.75\mn,1/4) {\size{$a$}};
		\node at (3.75\mn,-3.5/4) {\size{$b$}};
}} 
\end{equation}
and there are analogous diagrams (with additional strands in between, and to the left and right) for the generalized Clifford algebras with more generators.

We now observe that the lemma \ref{braidcommuting} can be reinterpreted in terms of respecting charge conservation, i.e. bringing an element of definite charge across the braid will \textbf{conserve} the charge, which is in this case just $0$. Thus, we say that the relation \ref{braidcommuting} provides a physical constraint on the action of the braid. In fact, this physical constraint provides a compelling explanation for why the master intertwining relation \ref{braideqn} holds; the latter is essentially forced by the constraint and the additional relation $b_{kl} c_l = c_k b_{kl}$.

\subsection{Applications of the Golden Rule}

Using the prior sections on the golden rule and various intertwining identities, we can now prove some identities involving the braid in a relatively straightforward manner. The following proof of unitarity is new, although the result is easily shown using explicit summation and is known \cite{Jaffe2017}. The importance of this new proof is that it introduces a new approach, using the trivial center property of the generalized Clifford algebra, which extends to proving identities for sums which are extremely difficult to calculate.
\subsubsection{Unitarity}
\begin{proposition}[Unitarity of Braid Elements]
	\label{braidunitarity}
	Suppose $|k-l|=1$, then
	\begin{equation}
    b_{kl} b_{lk}= b_{lk} b_{kl}=1.
    \end{equation}(As was remarked in the definition of the braids, $b_{kl}^{\dagger} = b_{lk}$, so equivalently, $b_{kl}$ is unitary.)
\end{proposition}
\begin{proof}
	Fix $k<l$, so we fix the braid elements. To prove this identity, we rely on propositions  \ref{identitycertificate1} and \ref{identitycertificate2}. Thus, we just need to show that a) $b_{kl} b_{lk}$ and $b_{lk}b_{kl}$ lie in the center, and b) the constant terms of $b_{kl} b_{lk}$ and $b_{lk} b_{kl}$ are both 1. To show that they lie in the center, we need to check that $c_p$ commutes with $b_{kl} b_{lk}$ for all $p$. 	Note that if $p<k<l$ or $p>l>k$, then $c_p$ commutes with $b_{kl}$  since it commutes with $c_k^a c_l^{-a}$ by lemma \ref{outercommuting}. We now note that $c_p b_{kl} = b_{kl} c_p$ implies the adjoint equation $b_{lk} c_p^{-1} =c_p^{-1} b_{lk}$, which further yields $b_{lk} c_p = c_p b_{lk}$ by iterating the commutation relation for $c_p^{-1}$ $N-1$ times. Thus, $c_p$ commutes with both $b_{kl}$ and $b_{lk}$.	Since $|k-l|=1$, the only other possibilities we need to check for $c_p$ are $p=k$ or $p=l$. 
	
	Recall that we have  the master braid identity \ref{braideqn}: $b_{kl} c_k^a c_l^b = q^{a^2+ab} c_k^{2a+b} c_l^{-a} b_{kl}$.  Applying this identity allows us to bring $c_k$ past $b_{kl}b_{lk}$ via 
    \begin{align}
         b_{kl} b_{lk} c_{k} &= b_{kl} c_l b_{lk} \\
        &= c_k b_{kl} b_{lk},
    \end{align} and $c_l$ past $b_{kl} b_{lk}$ via the slightly more involved
    \begin{align}
        b_{kl} b_{lk} c_{l} &= q\,b_{kl}  c_k^{-1} c_l^2 b_{lk} \\
        &= c_l b_{kl}b_{lk}.
    \end{align} Thus, $b_{kl} b_{lk}$ lies in the center. A similar argument using the adjoint braid identity, equation \ref{adjointbraideqn}, yields the computation  \begin{align}
        b_{lk} b_{kl} c_{l} &= b_{lk} c_k b_{kl} \\
        &= c_l b_{lk} b_{kl},
    \end{align} and 
    \begin{align}
        b_{lk} b_{kl} c_{k} &= q \, b_{lk} c_k^{2} c_l^{-1} b_{kl} \\
        &= c_k b_{lk} b_{kl},
    \end{align} so $b_{lk} b_{kl}$ lies in the center as well.
	
	We now need to compute the constant terms for $b_{kl} b_{lk}$ and $b_{lk} b_{kl}$. A direct computation \\ \noindent shows that $b_{kl} b_{lk}$ has the constant term $\frac{1}{N} \sum_{i=0}^{N-1} (c_k^i c_l^{-i})(c_l^{i} c_k^{-i}) = 1$. Similarly, $b_{lk} b_{kl}$ has the constant term $\frac{1}{N}\sum_{i=0}^{N-1} (c_l^i c_k^{-i})(c_k^i c_l^{-i})=1$. Thus, applying proposition \ref{identitycertificate1} in the case $x=b_{kl} b_{lk}$ and $y=1$, we obtain that $b_{kl} b_{lk} = 1$. Similarly, again applying proposition \ref{identitycertificate1} and setting $x =b_{lk} b_{kl} $ and $y=1$, we obtain that  $b_{lk} b_{kl}=1$, concluding the proof.
\end{proof}

The corresponding graphical identity for unitarity, for the special case $n=1$ (only two generators), $b_{21}b_{12}=b_{12}b_{21}$, is 
	\begin{equation}\raisebox{-0.5cm}{
	{\tikz{
			\fbraid{2\mn}{0}{0\mn}{2\nn}
			\fbraid{0\mn}{2\nn}{2\mn}{4\nn}
}}}\, =
	\!	\raisebox{-.4cm}{
	\tikz{
		\draw (1\mn,2)--(1\mn,1);
		\draw (0\mn,2)--(0\mn,1);	
}} \; \,\,.
\end{equation}
Analogous graphical identities hold for $b_{k,k+1}$ and for general $n$, where one puts more strands to the left and right of the above diagram. Again, we emphasize the requirement of having a diagram being represented by all strands. Hence, the above diagram does \textit{not} represent the unitarity condition for all $b_{kl}$, but merely for $b_{12}$. 

In fact, we can now generalize the above unitarity condition extends to braid elements with no graphical interpretation at all:
\begin{corollary}
	\begin{equation}
	b_{kl} b_{lk} = b_{lk} b_{kl} = 1
	\end{equation}
	for all $k\neq l$ in the set $\{1,2,\ldots, 2n\}$.
\end{corollary}
\begin{proof}
	Suppose without loss of generality that $k<l$, and consider the isomorphism of subalgebras $\langle c_1, c_2 \rangle$ and $\langle c_k, c_l \rangle$ given by the linear mapping $\phi$ satisfying $\phi(c_1^a c_2^b):= c_k^a c_l^b$, defining $\phi$ by its action on a basis for the subalgebra $\langle c_1, c_2 \rangle$. This is an isomorphism since \\ \noindent $\phi((c_1^a c_2^b) (c_1^i c_2^j)) = \phi(q^{-bi} c_1^{a+i} c_2^{b+j}) = q^{-bi} c_k^{a+i} c_l^{b+j} = c_k^a c_l^b c_k^i c_l^j = \phi(c_1^a c_2^b) \phi(c_1^i c_2^j)$, and the map is invertible.	By double distributivity of multiplication in the two subalgebras, the mapping extends to a homomorphism, and thus is an isomorphism. The isomorphism maps $b_{12} b_{21}$ to $b_{kl} b_{lk}$ and $1$ to $1$, so we obtain that $b_{kl} b_{lk} = 1$. Similarly, $b_{lk} b_{kl} = 1$.
	
\end{proof}

The above proof of proposition \ref{braidunitarity} may seem slightly over-kill, since we could have also expanded the product of $b_{kl}$ and $b_{lk}$, and performed the double sum. The strength (and elegance) of the method becomes more apparent when one deals with more complicated products, which is what we take up next.

\subsubsection{Yang-Baxter Equation and Braid Group Realization}
We now give one of our main results, which is an explicit algebraic proof of a Yang-Baxter equation, using the golden rule and a systematic application of the master braid and adjoint braid identities. The Yang-Baxter equation \cite{Yang1967} reads as $ABA = BAB$ and is what is known as a braid relation. More formally, we will establish the \textit{braid relations} satisfied by the braid group generated by the $b_{k,k+1}$'s. The braid group, introduced by Artin\cite{Artin1925}, is defined to be the object
\begin{equation}
    B_L = \langle \sigma_1, \ldots, \sigma_{L-1}|\sigma_k \sigma_{k+1} \sigma_k = \sigma_{k+1} \sigma_k \sigma_{k+1}, \sigma_k \sigma_l = \sigma_l \sigma_k \text{ if } |k-l|\geq 2 \rangle.
\end{equation}
We need to show that, setting $\sigma_{k} = b_{k,k+1}$ for $k=1,2,\cdots, 2n-1$, these $\sigma_k$'s satisfy the relations for the braid group generators.

We first present a proof of a special case of the Yang-Baxter equation, specialized to a generalized Clifford algebra with three generators $c_1, c_2, c_3$:
\begin{theorem}[Special Case of the Yang-Baxter Equation]
	\label{YBE}
	\begin{equation}
 b_{12} b_{23} b_{12} = b_{23} b_{12} b_{23}
 \end{equation}
\end{theorem}
\begin{proof}
	Since the braid elements are unitary, it suffices to prove the assertion that \\ \noindent $b_{32} b_{21} b_{32} b_{12} b_{23} b_{12}$ lies in the center and 
that the constant of proportionality between $b_{12} b_{23} b_{12}$ and $b_{23} b_{12} b_{23}$ is 1. By Proposition \ref{identitycertificate2}, to show that $b_{32} b_{21} b_{32} b_{12} b_{23} b_{12}$ lies in the center, we just need to show that it commutes with $c_k$ for all $k=1,2,\cdots, 2n$. Clearly, for $k>3$, \\ \noindent $b_{32} b_{21} b_{32} b_{12} b_{23} b_{12}$ commutes with $c_k$, since each braid element commutes with $c_k$. So we want to do case analysis for $k=1,2,3$. For $k=1$, 
\begin{align}
    b_{32} b_{21} b_{32} b_{12} b_{23} b_{12} c_1 &= q b_{32} b_{21} b_{32} b_{12} b_{23} c_1^2 c_2^{-1} b_{12}  \\
    &= q^2 b_{32} b_{21} b_{32} b_{12} c_1^2 c_2^{-2} c_3 b_{23} b_{12}\\
    &=q^2 b_{32} b_{21} b_{32} c_1^{2} c_2^{-2} c_3 b_{12} b_{23} b_{12}
\end{align} after applying the master braid identity, Proposition \ref{braideqn} thrice and using Lemma \ref{outercommuting}. Applying the adjoint braid identity thrice (equation \ref{adjointbraideqn}) then yields 
\begin{align}
    q^2 b_{32} b_{21} b_{32} c_1^{2} c_2^{-2} c_3 b_{12} b_{23} b_{12} &= q b_{32} b_{21} c_1^{2} c_2^{-1} b_{32} b_{12} b_{23} b_{12} \\
    &= b_{32} c_1 b_{21}  b_{32} b_{12} b_{23} b_{12} \\
    &= c_1 b_{32} b_{21}  b_{32} b_{12} b_{23} b_{12},
\end{align} as desired. The cases $k=2$, $k=3$ are similarly shown to satisfy 
\begin{equation}
    b_{32} b_{21} b_{32} b_{12} b_{23} b_{12} c_k = c_k b_{32} b_{21} b_{32} b_{12} b_{23} b_{12}
\end{equation} in like manner. Thus, we conclude that $b_{32} b_{21} b_{32} b_{12} b_{23} b_{12}$ lies in the center.

It remains to show that the constant of proportionality between $b_{12} b_{23} b_{12}$ and $b_{23} b_{12} b_{23}$ is 1. First focus on the constant terms. Since $b_{kl} = \frac{\omega^{1/2}}{\sqrt{N}} \sum_{i=0}^{N-1} c_k^i c_l^{-i}$, it suffices to compare the constant terms of 
$\sum_{i,j,k=0}^{N-1} (c_1^{i} c_2^{-i})(c_2^{j}c_3^{-j})(c_1^{k} c_2^{-k})$  and $\sum_{i,j,k=0}^{N-1} (c_2^{i} c_3^{-i})(c_1^{j}c_2^{-j})(c_2^{k} c_3^{-k})$. Note that in the first sum, the constant term only includes terms with $i+k=0$ and $j=0$, so the constant is given by $\sum_{i=0}^{N-1} (c_1^{i} c_2^{-i})(c_1^{-i} c_2^{i}) = \sum_{i=0}^{N-1} q^{-i^2}$. In the second sum, the constant term only includes terms with $j=0$ and $i+k=0$, so the constant is given by $\sum_{i=0}^{N-1} (c_2^{i} c_3^{-i})(c_2^{-i} c_3^{i})= \sum_{i=0}^{N-1}  q^{-i^2}$. Clearly the constant terms agree. However, this is not sufficient to conclude the constant of proportionality is 1, since the constant term may vanish. In fact, for $N=2(\text{mod }4)$, it does vanish, while it does not vanish for other $N$. This fact is due to the following formulas corresponding to Gauss' classical result for quadratic sums, which are tabulated in \cite{Hansen}:
\begin{equation}
\sum_{k=0}^{n-1} \sin\left(\frac{2\pi k^2}{n}\right)=\frac{\sqrt{n}}{2} \left(1+\cos(n\pi/2)-\sin(n\pi/2)\right)
\end{equation}

\begin{equation}
\sum_{k=0}^{n-1} \cos\left(\frac{2\pi k^2}{n}\right)=\frac{\sqrt{n}}{2} \left(1+\cos(n\pi/2)+\sin(n\pi/2)\right)
\end{equation}
Applying these formulas to $\sum_{i=0}^{N-1}  q^{-i^2} = \sum_{k=0}^{N-1} \exp{-2\pi i k^2/N}$ yields that the real part of the sum vanishes if $1+\cos(N\pi/2)+\sin(N\pi/2)$ vanishes, and the imaginary part vanishes if  $1+\cos(N\pi/2)-\sin(N\pi/2)$ vanishes. Thus, we require that $\cos(N\pi/2)=-1$ and $\sin(N\pi/2)=0$, so $N\pi/2 = \pi+2 m \pi$ and $N\pi/2 = l \pi$, i.e. $N= 2+4 m$ and $N=2l$, i.e. $N=2 (\text{mod } 4)$. This shows that the constant term does not vanish unless $N=2 (\text{mod } 4)$.

Now focus on the term with $c_2 c_3^{-1}$. In the first sum, this term is $\left(\sum_{i=0}^{N-1} q^{i-i^2} \right)c_2 c_3^{-1}$. In the second sum, this term is  $\sum_{i,k=0}^{N-1} (c_2^{i} c_3^{-i})(c_2^{1-i} c_3^{i-1}) = \left(\sum_{i=0}^{N-1}q^{i-i^2}\right)c_2 c_3^{-1}$, so the two terms are identical. The multiplicative factor $\sum_{i=0}^{N-1}q^{i-i^2} = q^{1/4} \sum_{k=0}^{N-1} q^{-(k-1/2)^2} $, which equals \\ \noindent $q^{1/4} \sum_{k=0}^{N-1} e^{-2\pi i(2k-1)^2/4N}$, vanishes only for $N=0$ (mod 4) by a result of Tseng \cite{Tseng}.

Thus, the constant term and the $c_2 c_3^{-1}$ term agree and their sum can never vanish. Hence, we conclude that the constant of proportionality must be 1, as desired.

\end{proof}

The corresponding graphical identity for the Yang-Baxter equation $b_{12} b_{23} b_{12} = b_{23}b_{12} b_{23}$ is given economically  for the algebra with 3 generators $c_1$, $c_2$, $c_3$, as

	\begin{equation}\raisebox{-0.9cm}
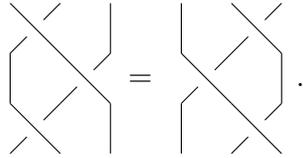
{
	{\tikz{
			\draw (-2/3,2/3)--(-2/3,4/3);
			\fbraid{0\mn}{0}{2\mn}{2\nn}
			\draw (2/3,4/3)--(2/3,6/3);
			\draw (2/3,0/3)--(2/3,2/3);
			\fbraid{-2\mn}{2\nn}{0\mn}{4\nn}
			\fbraid{0\mn}{4\nn}{2\mn}{6\nn}
}}}\;
=\; \raisebox{-0.9cm}{
	{\tikz{
			\draw (-4/3,4/3)--(-4/3,6/3);
			\draw (-4/3,0/3)--(-4/3,2/3);
			\fbraid{2\mn}{2\nn}{4\mn}{4\nn}
			\fbraid{0\mn}{0}{2\mn}{2\nn}
			\draw (0/3,2/3)--(0/3,4/3);
			\fbraid{0\mn}{4\nn}{2\mn}{6\nn}
}}}\;\,.\end{equation}
For $2n$ generators, one needs to put $2n-3$ strands to the right of the diagram for completeness.

Similar to the case of the unitarity condition, a more general Yang-Baxter-like equation holds for braid elements which do not admit a graphical interpretation:
\begin{theorem}[General Case of the Yang-Baxter Equation]
\label{generalYBE}
	Suppose $i<j<k$, then
	\begin{equation}
        b_{ij} b_{jk} b_{ij} = b_{jk} b_{ij} b_{jk}.
    \end{equation}
\end{theorem}
\begin{proof}
	We define an isomorphism, this time between the subalgebras $\langle c_1, c_2, c_3 \rangle$ and $\langle c_i, c_j, c_k \rangle$. Specifically, define $\phi$ by its action on a basis for the subalgebra $\langle c_1, c_2, c_3 \rangle$ via $\phi(c_1^p c_2^q c_3^r):=c_i^p c_j^q c_k^r$ for all $p,q,r \in \{0,1,\ldots,N-1\}$. Clearly, $\phi(1) =1$. Furthermore, $\phi$ is a homomorphism since \begin{align}
	    \phi((c_1^u c_2^v c_3^w)( c_1^p c_2^q c_3^r))&=\alpha\, \phi(c_1^{u+p} c_2^{v+q} c_3^{w+r}) \\
     &= \alpha \, c_i^{u+p} c_j^{v+q} c_k^{w+r} \\
     &= (c_i^u c_j^v c_k^w)(c_i^p c_j^q c_k^r),
     \end{align} where $\alpha$ collects all the phase factors from commuting the $c$'s around. It is clear that $\phi$ is a one-to-one mapping. Then applying $\phi$ to the product formula 
 \begin{equation}
	    b_{32} b_{21} b_{32} b_{12} b_{23} b_{12}=1
	\end{equation} yields 
 \begin{equation}
     b_{kj} b_{ji} b_{kj} b_{ij} b_{jk} b_{ij}=1,
\end{equation} which implies the desired result by taking the adjoint braids to the other side to become braids.
\end{proof}

Now we claim that setting $\sigma_{k}=b_{k,k+1}$ yields the desired braid group. 
\begin{theorem}
    \label{braidgroup}
    Set $\sigma_k = b_{k,k+1}$. These elements generate a unitary representation of the  braid group
    \begin{equation}
    B_{2n} = \langle \sigma_1, \ldots, \sigma_{2n-1}|\sigma_k \sigma_{k+1} \sigma_k = \sigma_{k+1} \sigma_k \sigma_{k+1}, \sigma_k \sigma_l = \sigma_l \sigma_k \text{ if } |k-l|\geq 2 \rangle.
\end{equation}
\end{theorem}
\begin{proof}
    The condition $\sigma_k \sigma_{k+1} \sigma_k = \sigma_{k+1} \sigma_k \sigma_{k+1}$ is true by Proposition \ref{generalYBE} taking the three generators to be $c_k,c_{k+1},c_{k+2}$. Meanwhile, the commutation relation $\sigma_k \sigma_l = \sigma_l \sigma_k$ for $|k-l|\geq 2$ follows by applying the linear extension of Proposition \ref{outercommuting}.
\end{proof}

\subsection{Significance of the Yang-Baxter Equation Proof}

At this point, we wish to elaborate on the significance of our algebraic proof of the Yang-Baxter equation. This subsection is divided into two parts, the first being the particular \textit{local}  representation for the $b_{k,k+1}$'s built out of $c_{i}$'s satisfying the two axioms, and the second being the local representation for an alternate local representation $b_{k,k+1}$'s built out of $c_{i}$'s not conforming to the explicit representation we constructed to satisfy our two axioms, but still satisfying the relations of a generalized Clifford algebra. By local, we mean that the unitary braid elements are 2-qudit entangling gates or single-qudit gates, in the terminology of quantum circuits; and furthermore, only adjacent qudits are entangled. Via a suitable realization of the generalized Clifford algebras, the latter section provides a solution to an open question in the work of Cobanera and Ortiz \cite{Cobanera2014}, regarding the construction of unitary solutions realizing the braid group $B_{2n}$ when the underlying qudit dimension $N$ of the $n$-qudit system is even, of the ``self-dual'' form:
\begin{align}
\rho_{sd}(\sigma_{2i-1}) &= \frac{1}{\sqrt{N}}\sum_{m=0}^{N-1} \alpha_m U_i^{-m}, i=1,\ldots, n \\
\rho_{sd}(\sigma_{2i}) &= \frac{1}{\sqrt{N}} \sum_{m=0}^{N-1} \beta_m V_i^m V_{i+1}^{-m}, i=1,\dots, n-1.
\end{align}
Here, the operators $V_k$ and $U_k$, termed Weyl generators, are defined by
  \begin{equation}
	V_k \ket{a_1,a_2,\ldots,a_n}=  \ket{a_1,a_2,\ldots, (a_k-1)(\text{mod }N),\ldots,a_n}
 \end{equation}
	and
	\begin{equation}
	   U_k \ket{a_1,a_2,\ldots,a_n}=  q^{a_k} \ket{a_1,a_2,\ldots, a_k,\ldots,a_n}.
	\end{equation}
 $V_k$ and $U_k$ satisfy the commutation relation $V_k U_k = q U_k V_k$ and Weyl generators with different $k$'s commute.  The operators $V_k$, $U_k$ correspond to the generalized Pauli operators $X^{-1}$ ($X$ is bit increment) and $Z$ ($Z$ is phase increment).

\subsubsection{Local Representation of the $b_{k,k+1}$'s} We first recall \cite{Lin1} the particular realization of the generalized Clifford algebras that was constructed in order to satisfy the two axioms:
  \begin{equation}
	c_{2k}\ket{a_1,a_2,\ldots,a_n}= q^{-\sum_{i<k} a_i} \ket{a_1,a_2,\ldots, (a_k+1)(\text{mod }N),\ldots,a_n}
 \end{equation}
	and
	\begin{equation}
	    c_{2k-1}\ket{a_1,a_2,\ldots,a_n}= \zeta \, q^{a_k} q^{-\sum_{i<k} a_i} \ket{a_1,a_2,\ldots, (a_k+1) (\text{mod }N),\ldots,a_n}.
	\end{equation}
 
 To connect to \cite{Cobanera2014}, we need to rewrite $c_{2k}$ and $c_{2k-1}$ in terms of the single-qudit generalized Pauli operators, also called Heisenberg-Weyl operators. Such rewriting in terms of single-qudit operators is known as \textit{a} Jordan-Wigner transformation \cite{Jaffe2017}; the particular Jordan-Wigner transformation depends on some conventions about phases and the single-qudit operators chosen and needs to be computed explicitly. Thus, there was some nontriviality in verifying the axioms we presented, since we insisted on particular phases associated with the corresponding $c_{2k}$ and $c_{2k-1}$'s in axiom 1, which depend in some way on the parity of $N$. 
 
 In our case, we compute the Jordan-Wigner transformation using the single-qudit operators of \cite{Cobanera2014}, $U_k$ and $V_k$ above. 
 Thus, 
 \begin{equation}
    \label{qudit2k}
     c_{2k} = U_1^{-1} U_2^{-1} \cdots U_{k-1}^{-1} V_k^{-1}
 \end{equation}
 and
 \begin{equation}
    \label{qudit2k-1}
     c_{2k-1} = \zeta   U_1^{-1} U_2^{-1} \cdots U_{k-1}^{-1} V_k^{-1} U_k.
 \end{equation}
 First, we show that $c_{2k-1} c_{2k}^{-1}$ is 1-local:
 \begin{proposition}
    \label{1local}
     $c_{2k-1} c_{2k}^{-1}$ is 1-local, i.e. it only acts on the $k$th qudit and leaves the rest fixed. In particular, $c_{2k-1} c_{2k}^{-1} =\zeta^{-1} U_k$.
 \end{proposition}
 \begin{proof}
     \begin{align}
          c_{2k-1} c_{2k}^{-1}&=\left(\zeta   U_1^{-1} U_2^{-1} \cdots U_{k-1}^{-1} V_k^{-1} U_k\right) \left(U_1 U_2 \cdots U_{k-1} V_k \right)
          \\
          &=\zeta V_k^{-1} U_k V_k \\
          &= \zeta q^{-1} V_k^{-1} V_k U_k  \\
          &=  \zeta^{-1} U_k.
     \end{align}
 \end{proof}
 It will be convenient also to have $c_{2k+1}$ and $c_{2k+1}^{-1}$ at our disposal:
 \begin{align}
     c_{2k+1} &= \zeta   U_1^{-1} U_2^{-1} \cdots U_{k-1}^{-1} U_k^{-1} V_{k+1}^{-1}  U_{k+1} \\ 
     \label{qudit2kplus1}
     c_{2k+1}^{-1} &= \zeta^{-1}   U_1 U_2 \cdots U_{k-1} U_k  U_{k+1}^{-1}V_{k+1}.
 \end{align}
 Thus, the following combination is 2-local:
 \begin{proposition}
    \label{twolocal}
     $c_{2k}c_{2k+1}^{-1}$ is 2-local, i.e. it only acts on the $k$th and $(k+1)$th qudits and leaves the rest of them fixed. In particular, 
     \begin{equation}
         c_{2k}c_{2k+1}^{-1} =  \zeta^{-1} V_k^{-1} U_k U_{k+1}^{-1} V_{k+1}.
     \end{equation}
 \end{proposition}
 \begin{proof}
     Using equations \ref{qudit2k} and \ref{qudit2kplus1}, 
     \begin{align}
     c_{2k}c_{2k+1}^{-1}&=\left(U_1^{-1} U_2^{-1} \cdots U_{k-1}^{-1} V_k^{-1}\right)\left(\zeta^{-1}   U_1 U_2 \cdots U_{k-1} U_k  U_{k+1}^{-1}V_{k+1}\right)\\
     &= \zeta^{-1} V_k^{-1} U_k U_{k+1}^{-1} V_{k+1}.\end{align}
     Since $U_k$, $V_k$ act only on the $k$th qudit, it follows that $c_{2k}c_{2k+1}^{-1}$ only acts on the $k$th and $(k+1)$th qudits.
 \end{proof}

As a consequence, we obtain the important relation that the braid elements $b_{2k,2k+1}$ are 2-local:
 \begin{theorem}
    $b_{2k,2k+1}$  is 2-local. In particular,
     \begin{align}
     b_{2k,2k+1} &=\frac{\omega^{1/2}}{\sqrt{N}} \sum_{i=0}^{N-1} \zeta^{-i^2} W_k^i W_{k+1}^{-i},
 \end{align}
 where $W_k = V_k^{-1} U_k$ for each $k\in\{1,2,\ldots, n\}$.
 \end{theorem}
 \begin{proof}
     
 Recall that \begin{equation}b_{kl}:=\frac{\omega^{1/2}}{\sqrt{N}} \sum_{i=0}^{N-1} c_k^i c_l^{-i} \end{equation} defines the braid elements. We will compute $b_{2k,2k+1}$ in terms of $U_k$, $V_k$, $U_{k+1}$ and $V_{k+1}$. 
 
 \begin{lemma}
    \label{gencombination}
     Suppose $c_k c_l = Q c_l c_k$, then $(c_k c_l^{-1})^{n}=Q^{n(n-1)/2} c_k^n c_l^{-n}$.
 \end{lemma} 
 \begin{proof}
 Suppose $c_k c_l = Q c_l c_k$, then 
 \begin{equation}c_k c_l^{-1} = c_k c_l^{N-1} =Q^{N-1} c_l^{N-1} c_k = Q^{-1} c_l^{-1} c_k\end{equation}. Thus, $c_k^n c_l^{-n}$ in terms of $(c_k c_l^{-1})^{n}$ is given by
 \begin{align}
     (c_k c_l^{-1})^{n} &= c_k c_l^{-1} c_k c_l^{-1} \cdots c_k c_l^{-1} \\ &= Q c_k^{2} c_l^{-2} c_k c_l^{-1} \cdots c_k c_l^{-1} \\
     &= Q^{1+ 2+ \cdots + (n-1)} c_k^{n} c_l^{-n} \\
     &= Q^{n(n-1)/2} c_k^n c_l^{-n}.
 \end{align}
 \end{proof}
 In particular, $c_{2k} c_{2k+1} = q c_{2k+1}c_{2k}$, so
 \begin{equation}
     c_{2k}^{n} c_{2k+1}^{-n}= q^{-n(n-1)/2}(c_{2k} c_{2k+1}^{-1})^n. 
 \end{equation}

 Thus, applying Proposition \ref{twolocal}
 \begin{align}
     b_{2k,2k+1} &= \frac{\omega^{1/2}}{\sqrt{N}} \sum_{i=0}^{N-1} q^{-i(i-1)/2} (c_{2k} c_{2k+1}^{-1})^i \\
     &=\frac{\omega^{1/2}}{\sqrt{N}} \sum_{i=0}^{N-1} q^{-i(i-1)/2} (\zeta^{-1} V_k^{-1} U_k U_{k+1}^{-1} V_{k+1})^i \\
    &=\frac{\omega^{1/2}}{\sqrt{N}} \sum_{i=0}^{N-1} q^{-i(i-1)/2} \zeta^{-i} (V_k^{-1} U_k)^{i} (U_{k+1}^{-1} V_{k+1})^i  
 \end{align}
 For convenience, set $W_k = V_k^{-1} U_k$ for each $k$, and rewrite $q=\zeta^2$, yielding
 \begin{align}
     b_{2k,2k+1} &=\frac{\omega^{1/2}}{\sqrt{N}} \sum_{i=0}^{N-1} \zeta^{-i(i-1)} \zeta^{-i} W_k^i W_{k+1}^{-i} \\
     &= \frac{\omega^{1/2}}{\sqrt{N}} \sum_{i=0}^{N-1} \zeta^{-i^2} W_k^i W_{k+1}^{-i}.
 \end{align}
 \end{proof}
 As a consistency check, let us show that this form of the sum for $b_{2k,2k+1}$ is invariant under shifting the index by $N$. The proof is nontrivial in this generalized Pauli basis, as it requires a cancellation of covariant factors. From a physics perspective, we remark that the cancellation of covariant factors is reminiscent of the construction of scalars in the theory of general relativity.
\begin{theorem}[Cancellation of Covariant Factors]
    Each term in the sum $b_{2k,2k+1}= \frac{\omega^{1/2}}{\sqrt{N}} \sum_{i=0}^{N-1} \zeta^{-i^2} W_k^i W_{k+1}^{-i}$ is invariant under shifting the sum index by $N$. Thus, the sum is invariant under shifting the indexing by arbitrary integers.
\end{theorem}
\begin{proof}
    Note that $W_k^N = -1$ if $N$ is even, since $V_k^N=U_k^N=1$, $V_k U_k  = q U_k V_k$ and we can apply Lemma \ref{gencombination} for $W_k = V_k^{-1} U_k$ to obtain that $W_k^N = Q^{N(N-1)/2}$. As $V_K^{-1} U_k = q^{-1} U_k V_k^{-1}$, it follows that $Q=q^{-1}$, so $W_k^N = q^{-N(N-1)/2}$. Since $q$ is a primitive $N$th root of unity, $q^{-N/2}=-1$, so $W_k^N = (-1)^{(N-1)}=-1$ if $N$ is even. This is not a problem for the invariance of the sum of the braid, under shifting the index, since there are \textit{two} $W$'s, a $W_k$ and a $W_{k+1}$, so under shifting by $N$, one acquires two factors of $-1$, which cancel each other out. 
    
    If $N$ is odd, the $W$ factors are invariant under shifting by $N$ since \begin{equation}
    W_k^{N}=Q^{N(N-1)/2}=(Q^N)^{(N-1)/2}=1
    \end{equation} since $(N-1)/2$ is an integer.  Recall that in both cases, $\zeta$ is a square root of $q$ such that $\zeta^{N^2}=1$ so $\zeta^{-i^2}$ is invariant under translations by $N$. So each term in the sum is invariant under shifting the sum index by $N$.

    Finally, it follows that shifting the indexing (e.g., from $0$ to $N-1$, to $1$ to $N$) by arbitrary integers preserves the entire sum, since we can simply maps the terms back into $\mathbb{Z}_N$ by subtracting from or adding to the index of the relevant terms appropriate multiples of $N$.
\end{proof}
 
It remains to compute the form of $b_{2k-1,2k}$, which is accomplished with the aid of Lemma \ref{gencombination} and Proposition \ref{1local}:
 \begin{theorem}
    $b_{2k-1, 2k}$ is 1-local. In particular,
     \begin{equation}
         b_{2k-1,2k}= \frac{\omega^{1/2}}{\sqrt{N}} \sum_{i=0}^{N-1} \zeta^{-i^2} U_k^{i}
     \end{equation}
 \end{theorem}
 \begin{proof}
 Applying Lemma \ref{gencombination} and Proposition \ref{1local}:
 \begin{align}
      b_{2k-1,2k} = &= \frac{\omega^{1/2}}{\sqrt{N}} \sum_{i=0}^{N-1} q^{-i(i-1)/2} (c_{2k-1} c_{2k}^{-1})^i \\
      &= \frac{\omega^{1/2}}{\sqrt{N}} \sum_{i=0}^{N-1} q^{-i(i-1)/2} \left(\zeta^{-1} U_k\right)^{i} \\
      &= \frac{\omega^{1/2}}{\sqrt{N}} \sum_{i=0}^{N-1} q^{-i(i-1)/2} \left(\zeta^{-1} U_k\right)^{i} \\
      &= \frac{\omega^{1/2}}{\sqrt{N}} \sum_{i=0}^{N-1} \zeta^{-i(i-1)} \zeta^{-i} U_k^{i} \\
      &= \frac{\omega^{1/2}}{\sqrt{N}} \sum_{i=0}^{N-1} \zeta^{-i^2} U_k^{i}.
\end{align}
 \end{proof}
Note that the form of the braid group generators $b_{2k,2k+1}$ is \textit{not} in the requisite form of \cite{Cobanera2014} (one may neglect the unimodular phase factor $\omega$ in this comparison). It is, however, sufficiently similar, if one replaces $V$'s by $W$'s, that one expects that some adaptation of our approach should work to get solutions in the form desired by \cite{Cobanera2014}. We take up this problem next.

\subsubsection{A General Solution to the Open Question of Cobanera and Ortiz}
\label{openQuestion}

We now solve for braid elements of ``self-dual'' form given in \cite{Cobanera2014}:
\begin{align}
\rho_{sd}(\sigma_{2i-1}) &= \frac{1}{\sqrt{N}}\sum_{m=0}^{N-1} \alpha_m U_i^{- m}, i=1,\ldots, n \\
\rho_{sd}(\sigma_{2i}) &= \frac{1}{\sqrt{N}} \sum_{m=0}^{N-1} \beta_m V_i^{m} V_{i+1}^{-m}, i=1,\dots, n-1.
\end{align}
Our construction of a realization of the braid group $B_{2n}$ out of solutions of the self-dual form will depend on constructing a generalized Clifford algebra out of a particular combination of $U_k$'s and $V_k$'s. We will need to verify that the resulting particular Jordan-Wigner transformation from $U_k$'s and $V_k$'s indeed satisfies the relations of a generalized Clifford algebra. This verification step is a nontrivial point. In fact, in the original work of \cite{Cobanera2014}, the Jordan-Wigner transformation presented, expressing their generators $\Gamma_i$ and $\Delta_i$ (similar to our $c_{2k-1}$ and $c_{2k}$'s) in terms of the $U_i$'s and $V_i$'s, is incorrect. In odd qudit dimension, they were able to use results of Goldschmidt and Jones (see \cite{GJ1989} \cite{Jones1989}, namely equation 7-6) on braid group representations when $N$ is a power of an odd prime $p$,  to find a solution of the self-dual form. The flaw is that for \textit{even} qudit dimension, their $\Delta_i$ generators do not satisfy $\Delta_i^N=1$! The solution, informed by our development of our algebraic framework, is to incorporate the factor of $\zeta$ (appearing in our axiom $1$) to modify their Jordan-Wigner transformation. Thus, our construction illustrates once more the importance of the axiomatic approach \cite{Lin1} we are following, in which we both isolated the necessary algebraic structure in the two axioms, which depended on the choice of $\zeta$, and justified the validity of the two axioms by an explicit construction\footnote{As a reminder, $\zeta$ is a square root of $q$ such that $\zeta^{N^2}=1$, which guarantees that $\zeta^{-i^2}$ is invariant under shifting $i$ by $N$.}. Note that since for $N$ even, $\zeta$ can have two possible values, our construction gives rise to two distinct classes of solutions of the self-dual form.

Our starting point is Proposition \ref{braidgroup}, which asserts that the $b_{k,k+1}$'s constructed out of the generators $c_{i}$, for $i=1,2,\ldots, 2n$, generate the braid group $B_{2n}$. Since this proof only depends on the properties of the generalized Clifford algebra, rather than on a particular representation of the algebra, the proof extends to any construction of generators $c_{1}, c_2, \ldots, c_{2n-1}, c_{2n}$ out of the Weyl generators $U_j$ and $V_j$, which satisfies the relations of the generalized Clifford algebra, namely:
\begin{align}
    c_a c_b &= q c_b c_a \text{ if } a<b \\
    c_a^N &= 1 \, \text{ for any } a=1,2,\ldots, 2n.
\end{align}
In the following proposition, we  construct an automorphism of the generalized Clifford algebra which gives the mapping into the ``self-dual'' form specified by \cite{Cobanera2014}. 
We claim that using 
\begin{equation}
    u_{2k-1} = c_{2k}^{-1}
\end{equation} 
\begin{align}
u_{2k} = \zeta c_{2k}^{-1} U_k
\end{align}
yields an automorphism.  Since $U_k = \zeta c_{2k-1} c_{2k}^{-1}$, and phases that are powers of $q$ do not affect the GCA relations, we can alternately use the mapping
\begin{align}
    u_{2k-1} &= c_{2k}^{-1} \\
    u_{2k} &= c_{2k-1} c_{2k}^{-2}
\end{align}

\begin{proposition}
Define $u_{a}$ for $a=1,2,\ldots, 2n$ by
\begin{align}
    u_{2k-1} &= c_{2k}^{-1} \\
    u_{2k} &=  c_{2k-1} c_{2k}^{-2}
\end{align}
Then $u_a$ satisfies the relations of a generalized Clifford algebra, namely: 
\begin{align}
    u_a u_b &= q u_b u_a \text{ if } a<b \\
    u_a^N &= 1 \, \text{ for any } a=1,2,\ldots, 2n.
\end{align}

\end{proposition}
\begin{proof}
    By Lemma \ref{ABlemma}, two elements $x,y$ of charge $-1$, where $x$ is located on generators (graphically, strands) which are left of all the generators (strands) on which $y$ is located, commute past each other with $xy = q yx$, hence $u_{a}u_{b} = q u_{b} u_{a}$ for $a\in \{2k-1,2k\}$ and $b\in\{2l-1,2l\}$, $k<l$. So we simply need to check the commutation relation for $u_{2k-1}$ and $u_{2k}$. 
    \begin{align}
        u_{2k-1} u_{2k} &= c_{2k}^{-1} c_{2k-1} c_{2k}^{-2} = q c_{2k-1} c_{2k}^{-1} c_{2k}^{-2} \\
        &= q u_{2k} u_{2k-1}.
    \end{align}
    Furthermore, 
    \begin{align}
        u_{2k-1}^{N} &= c_{2k}^{-N} = 1\\
        u_{2k}^N &= \left( c_{2k-1} c_{2k}^{-2} \right)^N = Q^{N(N-1)/2} c_{2k-1}^N c_{2k}^{-2N} 
    \end{align}
    by Lemma \ref{gencombination}, where $c_{2k-1} c_{2k}^{-2} = Q c_{2k}^{-2} c_{2k-1}$. It is clear that $Q = q^{-2}$, hence $Q^{N(N-1)/2} = q^{-N(N-1)}= 1$. Thus,
    \begin{equation}
        u_{2k}^N = 1.
    \end{equation}
    Hence we have obtained an automorphism of the generalized Clifford algebra.
   
\end{proof}
\noindent \textbf{Remark:} Note that since one can construct $c_{2k-1}$ and $c_{2k}$ out of products of $u_{2k-1}$ and $u_{2k}$ and their powers and inverses, the size of the basis of the algebra is the same. This is a useful check to see whether the automorphism is actually an automorphism, independently of the relations. 

\begin{theorem}[Braid Group Representation]
    Define $\beta_{k,l}$ by 
    \begin{equation}
        \beta_{k,l} = \frac{1}{\sqrt{N}} \sum_{i=0}^{N-1} u_k^i u_l^{-i},
    \end{equation}
    where $u_a$ are as above. Then setting $\sigma_{k} = \beta_{k,k+1}$ for $k=1,2,\ldots, 2n-1$ yields a unitary representation of the braid group $B_{2n}$.
\end{theorem}
\begin{proof}
    Unitarity follows from the fact Proposition \ref{braidunitarity} only depends on the relations of the generalized Clifford algebra. Meanwhile, the braid group relations follow from the fact that the proof for Proposition \ref{braidgroup}, relying on the proof of the Yang-Baxter equation, and the commutation of elements of neutral charge, only depends on the properties of the generalized Clifford algebra as an algebra. Thus, we pass from $c_a$ to $u_a$ and Proposition \ref{braidgroup} still holds. Finally, since there is freedom in the definition of the braid element by a complex phase factor, we may change $\omega$ to $1$ without affecting unitarity.
\end{proof}
\begin{corollary}
    \label{morebraids}
    More generally, by the same proof, any automorphism of the generalized Clifford algebra will preserve unitarity as well as the braid group relations.
\end{corollary}
It remains to express the $\beta_{k,k+1}$'s in terms of the Weyl generators $V_i$,$U_i$.
\begin{theorem}
    $\beta_{2k-1,2k}$ is 1-local and $\beta_{2k,2k+1}$ is 2-local. They are given by
    \begin{align}
     \beta_{2k-1,2k} &= \frac{\zeta}{\sqrt{N}} \sum_{i=0}^{N-1} \zeta^{-(i-1)^2} U_k^{-i} \,\,\text{   for } k=1,2,\ldots, n        \\
        \beta_{2k,2k+1} &= \frac{\zeta}{\sqrt{N}} \sum_{i=0}^{N-1} \zeta^{-(i+1)^2} V_k^i V_{k+1}^{-i}  \,\,\text{   for } k=1,2,\ldots, n-1       
    \end{align}
\end{theorem}
\begin{proof}
Applying Lemma \ref{gencombination}:
\begin{align}
     \beta_{2k-1,2k} &= \frac{1}{\sqrt{N}} \sum_{i=0}^{N-1} q^{-i(i-1)/2} (u_{2k-1} u_{2k}^{-1})^i \\
     &= \frac{1}{\sqrt{N}} \sum_{i=0}^{N-1} q^{-i(i-1)/2} (c_{2k}^{-1} (c_{2k-1} c_{2k}^{-2})^{-1})^i \\
     &= \frac{1}{\sqrt{N}} \sum_{i=0}^{N-1} q^{-i(i-1)/2} (c_{2k}^{-1}  c_{2k}^2 c_{2k-1}^{-1} )^i \\
     &= \frac{1}{\sqrt{N}} \sum_{i=0}^{N-1} q^{-i(i-1)/2} (c_{2k} c_{2k-1}^{-1} )^i \\
     &= \frac{1}{\sqrt{N}} \sum_{i=0}^{N-1} q^{-i(i-1)/2} (c_{2k-1} c_{2k}^{-1} )^{-i}\\
    &=\frac{1}{\sqrt{N}} \sum_{i=0}^{N-1} q^{-i(i-1)/2} (\zeta^{-1} U_k )^{-i} \\
    &= \frac{1}{\sqrt{N}} \sum_{i=0}^{N-1} \zeta^{-i(i-1)} \zeta^{i} U_k^{-i}\\
    &= \frac{\zeta}{\sqrt{N}} \sum_{i=0}^{N-1} \zeta^{-(i-1)^2} U_k^{-i}
\end{align}
where we applied Proposition \ref{1local} to simplify $c_{2k-1} c_{2k}^{-1}$.

Applying Lemma \ref{gencombination} again:
 \begin{align}
     \beta_{2k,2k+1} &= \frac{1}{\sqrt{N}} \sum_{i=0}^{N-1} q^{-i(i-1)/2} (u_{2k} u_{2k+1}^{-1})^i \\
     &= \frac{1}{\sqrt{N}} \sum_{i=0}^{N-1} q^{-i(i-1)/2} (c_{2k-1} c_{2k}^{-2} (c_{2k+2}^{-1})^{-1})^i \\
     &= \frac{1}{\sqrt{N}} \sum_{i=0}^{N-1} q^{-i(i-1)/2} (c_{2k-1} c_{2k}^{-2} c_{2k+2})^i \\     
     &=\frac{1}{\sqrt{N}} \sum_{i=0}^{N-1} q^{-i(i-1)/2} ((\zeta   U_1^{-1} U_2^{-1}  \cdots U_{k-1}^{-1} V_k^{-1} U_k) \cdot (U_1^{-1} U_2^{-1} \cdots U_{k-1}^{-1} V_k^{-1})^{-2} \\
     & \, \,\,\, \cdot (U_1^{-1} U_2^{-1} \cdots U_{k-1}^{-1} U_{k}^{-1} V_{k+1}^{-1}))^i \\
     &=\frac{1}{\sqrt{N}} \sum_{i=0}^{N-1} q^{-i(i-1)/2} \zeta^i \left( V_k^{-1} U_k V_k^{2} U_k^{-1} V_{k+1}^{-1}\right)^i \\
     &= \frac{1}{\sqrt{N}} \sum_{i=0}^{N-1} q^{-i(i-1)/2} \zeta^i \left(q^{-2} V_k V_{k+1}^{-1}\right)^i \\
     &= \frac{1}{\sqrt{N}} \sum_{i=0}^{N-1} q^{-i(i-1)/2} \zeta^i q^{-2i}  V_k^i V_{k+1}^{-i}\\
     &= \frac{1}{\sqrt{N}} \sum_{i=0}^{N-1} \zeta^{-i(i-1)} \zeta^i \zeta^{-4i}  V_k^i V_{k+1}^{-i} \\
     &= \frac{\zeta}{\sqrt{N}} \sum_{i=0}^{N-1} \zeta^{-(i+1)^2} V_k^i V_{k+1}^{-i}.
 \end{align}
\end{proof}
In the braid elements, the indexing of the coefficients $\zeta^{-(i-1)^2}$ and $\zeta^{-(i+1)^2}$ is quite curious. Partially inspired by the suggestion of Cobanera and Ortiz \cite{Cobanera2014} that there may be many classes of braid group solutions of the self-dual form, we may try to extrapolate the coefficient to have different indexing. In particular, we may use the fact that the relations of the generators forming the generalized Clifford algebra are preserved under the scaling of generators $c_a$ and $c_b$ by factors of $q$ to generate different coefficients in the self-dual solutions. This appears to be related to a choice of \textbf{gauge} on each generator. Let us define $w_a(r_1,r_2,\ldots, r_{2n})$ by
\begin{align}
    w_a = q^{r_a} u_{a},
\end{align}
where $r_a \in \mathbb{Z}_N$. Then the $w_a$'s again form a generalized Clifford algebra. 
Then the new braid elements $\gamma_{k,k+1}$ are given by the following proposition:
\begin{proposition}
    \begin{align}
        \gamma_{2k-1,2k+1} &= \frac{\zeta^{(r_{2k}-r_{2k-1}-1)^2}}{\sqrt{N}} \sum_{i=0}^{N-1}  \zeta^{-(i+ (r_{2k}-r_{2k-1}-1))^2} U_k^{-i} \,\,\text{   for } k=1,2,\ldots, n      \\
        \gamma_{2k,2k+1} &= \frac{\zeta^{(1+r_{2k+1}-r_{2k})^2}}{\sqrt{N}}\sum_{i=0}^{N-1}  \zeta^{-(i+(1+r_{2k+1}-r_{2k}))^2} V_k^i V_{k+1}^{-i} \,\,\text{   for } k=1,2,\ldots, n-1      .
    \end{align}
\end{proposition}
\begin{proof}
We simply need to add in the rescaling factors induced in by the rescaling of the generators by phase factors:
\begin{align}
\gamma_{2k-1,2k} &= \frac{\zeta}{\sqrt{N}} \sum_{i=0}^{N-1} (q^{r_{2k-1}} q^{-r_{2k}})^i  \zeta^{-(i-1)^2} U_k^{-i}\\
&= \frac{1}{\sqrt{N}} \sum_{i=0}^{N-1} \zeta^{2(r_{2k-1}-r_{2k})i}  \zeta^{-i^2 +2i} U_k^{-i}\\
&=\frac{1}{\sqrt{N}} \sum_{i=0}^{N-1}  \zeta^{-(i^2 + 2(r_{2k}-r_{2k-1}-1)i)} U_k^{-i}\\
&=\frac{\zeta^{(r_{2k}-r_{2k-1}-1)^2}}{\sqrt{N}} \sum_{i=0}^{N-1}  \zeta^{-(i+ (r_{2k}-r_{2k-1}-1))^2} U_k^{-i}.
\end{align}
\begin{align}
\gamma_{2k,2k+1} &= \frac{\zeta}{\sqrt{N}} \sum_{i=0}^{N-1} (q^{r_{2k}} q^{-r_{2k+1}})^i \zeta^{-(i+1)^2} V_k^i V_{k+1}^{-i} \\
&= \frac{1}{\sqrt{N}} \sum_{i=0}^{N-1} \zeta^{2(r_{2k}-r_{2k+1})i} \zeta^{-i^2-2i} V_k^i V_{k+1}^{-i}\\
&= \frac{1}{\sqrt{N}} \sum_{i=0}^{N-1}  \zeta^{-(i^2+2(1+r_{2k+1}-r_{2k})i)} V_k^i V_{k+1}^{-i}\\
&= \frac{\zeta^{(1+r_{2k+1}-r_{2k})^2}}{\sqrt{N}}\sum_{i=0}^{N-1}  \zeta^{-(i+(1+r_{2k+1}-r_{2k}))^2} V_k^i V_{k+1}^{-i}.
\end{align}
\end{proof}
\begin{proposition}
    Setting $\sigma_k = \gamma_{k,k+1}$ yields a unitary braid group representation.
\end{proposition}
\begin{proof}
    The proposition follows by Corollary \ref{morebraids}.
\end{proof}
Since the phase of each braid element does not affect the braid group relations, it follows that up to phase, the set of self-dual braid group solutions that we have obtained is indexed by a $2n$-dimensional vector $(r_1,r_2, \ldots, r_{2n})$ in $\mathbb{Z}_N^{2n}$. Thus, using a particular \textit{automorphism} of the generalized Clifford algebra and the \textit{gauge symmetry} for each generator of the generalized Clifford algebra, we have obtained, from our proof of the Yang-Baxter equation and the related braid group construction, a general set of solutions to the braid group satisfying the ``self-dual'' form of Cobanera and Ortiz \cite{Cobanera2014}, which works for both odd and even $N$ ($N\geq 2$).

From a quantum computation standpoint, the braid elements are 2-local, and hence it is feasible that one might try to implement these gates. In fact, from the commutation relations \ref{braideqn} between the braid elements and the elements $c_a$, and the representation of $c_a$'s in terms of the generalized Pauli operators $V_k$ and $U_k$ from equations \ref{qudit2k} and \ref{qudit2k-1}, it is further evident that they \textit{almost} normalize the generalized Pauli group on $n$ qudits, the \textit{almost} being due to the extra factor of $\zeta$. To see this, simply examine the equation $b_{12} c_1=q c_{1}^{2} c_2^{-1} b_{12}$; $c_1$ has a prefactor $\zeta$, but $c_1^2$ has a prefactor of $q$, so the $\zeta$ factor remains. Further, observe that we may recover $V_k$ in terms of $\zeta$'s and the generalized Clifford algebra by using the expression for $c_{2k}$ in terms of $U_i$'s and the expression for $U_i$ in terms of $c_{a}$'s. Thus, we can access the entire generalized Pauli group, which is generated by $V_k$ and $U_k$'s, by appropriate products of generators of the generalized Clifford algebra, combined with appropriate factors of $\zeta$ ($q$ is contained in the generalized Pauli group, so it would be redundant to keep track of factors of $q$). Since these products of $c_a$'s can be commuted past the braid elements to yield again products of $c_a$'s time powers of $q$, it follows from the representation of any generalized Pauli operator as a product of  generators of the algebra up to powers of $\zeta$ that these braid elements are \textit{almost} Clifford gates, where the Clifford group \cite{Gottesman_1999} refers to the normalizer of the generalized Pauli group within the special unitary group over $n$ qudits of dimension $N$. 

\section{Graphical Calculus at the Level of the Multi-Qudit State}

The fact that the Yang-Baxter equation holds for the elements $b_{kl}$ of the generalized Clifford algebra suggests that perhaps some kind of identities should also hold for the \textit{vectors} with respect to the action of the generalized Clifford algebra. While one might speculate that the vectors (caps and cups) automatically satisfy a kind of an isotopy invariance, taking this to be a built-in axiom (in, e.g.,  \cite{Jaffe2017}) would most certainly be incompatible with the \textit{algebraic} axiomatic approach we have taken. Any such property ought to be \textit{derived} from the axioms we have presented, not simply taken to be true. Of course, when working with our vectors, we must stick to the representation we have chosen for the generalized Clifford algebra, so our investigation will by necessity proceed from axiom 1 of our algebraic framework.

To those who are familiar with some subfactor theory or category theory, it may be tempting to appeal to these theories as a kind of panacea for isotopy invariance with respect to braidings. However, it must be pointed out that one \textit{cannot} rely on the algebraic results of subfactor theory\footnote{Popa's results on the axiomatization of the standard invariant \cite{Popa1995} are for subfactors; one would need a (conjectural) graded subfactor theory, as noted in \cite{Jaffe2017}.} or tensor category theory\footnote{There \textit{is} no tensor category here, since the tensor product is not defined between two nonneutral elements of the generalized Clifford algebra. See, e.g., \cite{Muger}, for a nice exposition of tensor category theory.} approaches for any $N\geq 2$. 
In fact, our algebraic framework was devised precisely to enable one to circumvent these theoretical difficulties.

As the methods of proof we developed within the \textit{algebra} in the previous section cannot logically extend to proofs for the \textit{vectors}, we are forced to devise new methods to prove \textit{vector identities}. These methods are independent of the Yang-Baxter equation. It turns out that the results we obtain using these methods include not only graphical identities, but also encompass more general algebraic identities which supersede the graphical identities. In terms of our results, we will show that in a \textit{combinatorial} sense, two basic vector identities give rise to a plethora of identifications between different vectors generated from the ground state by braidings.

 First, we begin by proving a general projection-braid identity and two basic vector identities which uniformly apply to a multi-qudit space of an arbitrary number of qudits. The second vector identity, which we call the ``slip'' move, appears to be new. In their full generality, our two vector identities go beyond a graphical representation. We then show by example that these identities can be thought of as representing \textit{combinatorial} moves that one can perform on braided states without changing the state. We conclude with an example in which we show, without topological assumptions, 
 that two entangled vector states can be shown to be equal using these combinatorial moves in combination. 
 
 Thus, an important general result in this section is the introduction of a \textit{reduction procedure}: in many cases, one may reduce the problem of showing equivalence of two different sequences of braidings applied to the ground state, to that of a tractable combinatorial problem, instead of one of explicit algebraic computation. The essential starting point for these vector identities is the identity lemma \ref{caplemma}, and can be thought of as an important reason for using axiom 1 as an axiomatic starting point for the entire theory\footnote{
 Given how the ``rest'' of the theory is following from the axiomatic framework, the reader perhaps is gaining more appreciation of why it was so important to separate the algebraic framework into two parts: axioms which allow one to do lots of derivations and algebraic proofs, and a proof of that these axioms are satisfied by an explicit example, i.e. the existence of a consistent vector representation of the generalized Clifford algebra that satisfied both axiom 1 and axiom 2. The division of labor is made clear, and thus each part can be independently rigorously verified.}.

We start with the two main combinatorial moves we will need. In this section, as a matter of form, we will draw the diagrams first, and then writing out the algebraic expressions, as the diagrams in the vector representation take on increasing importance for intuition.
\begin{proposition}[Projection-Braid Identity, or the ``Twist'' Move]
	\label{Twist}
	\begin{equation}\raisebox{-0.75cm}{
		{\tikz{
				\fqudit {0}{4\nn}{1\mn}{.5\nn}{\phantom{ll}}
				\fmeasure {0}{7.5\nn}{1\mn}{.5\nn}{}
				\fbraid{-2\mn}{2\nn}{0\mn}{4\nn}
	}}}\, 
	=
	\omega^{-1/2}\raisebox{-.5cm}{
		\tikz{
			\fqudit {0}{0\nn}{1\mn}{.5\nn}{\phantom{ll}}
			\fmeasure {0}{3.5\nn}{1\mn}{.5\nn}{}
	}}
	\;
	\end{equation}
	
	Equivalently (by scaling the graphical identity by $\delta$), 
 \begin{equation}
 b_{12} E_1 = \omega^{-1/2} E_1.
 \end{equation}
	
	More generally,
	\begin{equation}
 b_{2k-1,2k} E_{k} = \omega^{-1/2} E_k
 \end{equation}
		for $k=1,2,\ldots, n$.
	\end{proposition}
\begin{proof}
	By definition, 
 \begin{align}
     b_{12}E_1 \, =\frac{\omega^{1/2}}{\sqrt{N}}\sum_{i=0}^{N-1}c_1^ic_2^{-i}E_1 .\end{align} Recall that the axioms for the projectors imply via lemma \ref{caplemma} that $c_1^a E_1 =\zeta^{a^2} c_2^a E_1$. So the above equality translates to
     \begin{align}
     b_{12} E_1 &=\frac{\omega^{1/2}}{\sqrt{N}}\left(\sum_{i=0}^{N-1}\zeta^{-i^2}\right)E_1 \\
     &=\omega^{1/2}\omega^*E_1 =\omega^{-1/2}E_1.
     \end{align}The general statement $b_{2k-1,2k} E_k = \omega^{-1/2} E_k$ follows similarly since the same lemma gives $c_{2k-1}^a E_k = \zeta^{a^2} c_{2k}^a E_k$, which allows for a similar simplification from the sum over generators to a single complex number.
\end{proof}

\begin{proposition}[``Slide'' Move]
	\label{overlap}
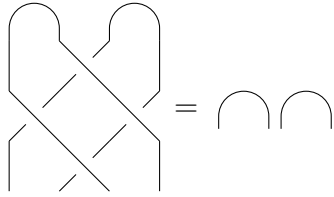
		\begin{equation}\raisebox{-1cm}{
		{\tikz{
				\fqudit {0}{6\nn}{1\mn}{.5\nn}{\phantom{ll}}
				\fqudit {-1.333}{6\nn}{1\mn}{.5\nn}{\phantom{ll}}
				\draw (-4/3,4/3)--(-4/3,6/3);
				\draw (-4/3,0/3)--(-4/3,2/3);
				\fbraid{2\mn}{2\nn}{4\mn}{4\nn}
				\fbraid{0\mn}{0}{2\mn}{2\nn}
				\draw (2/3,4/3)--(2/3,6/3);
				\draw (2/3,0/3)--(2/3,2/3);
				\fbraid{-2\mn}{2\nn}{0\mn}{4\nn}
				\fbraid{0\mn}{4\nn}{2\mn}{6\nn}
	}}}\, 
	=
	\raisebox{-.25cm}{
		\tikz{
			\fqudit {0}{0\nn}{1\mn}{.5\nn}{\phantom{ll}}
	}}
	\raisebox{-.25cm}{
		\tikz{
			\fqudit {0}{0\nn}{1\mn}{.5\nn}{\phantom{ll}}
	}}\;
	\end{equation}
	
	More generally (i.e. for $n$ (where $2n$ is the number of strands) not necessarily equal to 2), 
    \begin{equation}
     b_{23} b_{34}b_{12}b_{23}\ket{\Omega}^{\otimes n}=\ket{\Omega}^{\otimes n}.
    \end{equation}
\end{proposition}	
\begin{proof}
	Graphically, it is wisest to expand the braids on the 2nd and 3rd strands, since we may use existing algebraic graphical identities to simplify the result. This yields	
	\begin{equation}
 b_{23} b_{34}b_{12}b_{23}\ket{\Omega}^{\otimes n}=\frac{\omega}{N} \sum_{i,j=0}^{N-1} c_2^{j} c_3^{-j}  b_{34} b_{12} c_2^i c_3^{-i}\ket{\Omega}^{\otimes n}.
    \end{equation}
	Note that $b_{12}$, $b_{34}$ commute by linear extension of lemma \ref{outercommuting} so the order doesn't matter.
	
	In terms of a diagram, 	expanding the middle braids yields
	\begin{equation}
\frac{\omega}{N}\sum_{i,j=0}^{N-1}
	\raisebox{-1cm}{
	{\tikz{
			\fqudit {0}{4\nn}{1\mn}{1.5\nn}{\phantom{ll}}
				\draw (2/3,0/3)--(2/3,2/3);
				\draw (0/3,0/3)--(0/3,2/3);
			\fbraid{-2\mn}{2\nn}{0\mn}{4\nn}
					\node at (-1.5\mn,1/4) {\size{$j$}};
			\node at (-1.5\mn,6/4) {\size{$i$}};
}}} 
	\raisebox{-1cm}{
	{\tikz{
			\fqudit {0}{4\nn}{1\mn}{1.5\nn}{\phantom{ll}}
			\draw (2/3,0/3)--(2/3,2/3);
			\draw (0/3,0/3)--(0/3,2/3);
			\fbraid{-2\mn}{2\nn}{0\mn}{4\nn}
			\node at (1\mn,1.5/4) {\size{$-j$}};
			\node at (1\mn,6.5/4) {\size{$-i$}};
}}}\,= 
\frac{\omega}{N}\sum_{i,j=0}^{N-1}\zeta^{i^2} \!
\raisebox{-1.15cm}{
	{\tikz{
			\fqudit {0}{4\nn}{1\mn}{0.5\nn}{\phantom{ll}}
			\draw (2/3,-1/3)--(2/3,2/3);
			\draw (0/3,-1/3)--(0/3,2/3);
			\fbraid{-2\mn}{2\nn}{0\mn}{4\nn}
			\node at (-1.5\mn,-1/4) {\size{$j$}};
			\node at (0.5\mn,0.5/4) {\size{$i$}};
}}} \!
\raisebox{-1cm}{
	{\tikz{
			\fqudit {0}{4\nn}{1\mn}{0.5\nn}{\phantom{ll}}
		\draw (2/3,-1/3)--(2/3,2/3);
		\draw (0/3,-1/3)--(0/3,2/3);
		\fbraid{-2\mn}{2\nn}{0\mn}{4\nn}
		\node at (1\mn,-0.5/4) {\size{$-j$}};
		\node at (1\mn,1.5/4) {\size{$-i$}};
}}},
\end{equation} where we have applied axiom 1 to bring the charge $-i$ over to the 4th strand, yielding the phase factor $\zeta^{i^2}$, and then commuted it over the braid back to the 3rd strand. Similarly, the charge $i$ can be brought over the braid. Note that no additional phase accumulates, since overall the relative vertical positions of the charges are unchanged. Now apply the twist move in proposition \ref{Twist} to get the diagram
\begin{equation}
\frac{1}{N}\sum_{i,j=0}^{N-1}\zeta^{i^2} \!\!
\raisebox{-0.75cm}{
	{\tikz{
			\fqudit {0}{1\nn}{1\mn}{3.5\nn}{\phantom{ll}}
			\node at (-1.5\mn,2/4) {\size{$j$}};
			\node at (0.5\mn,4.5/4) {\size{$i$}};
}}} \!
\raisebox{-0.6cm}{
	{\tikz{
			\fqudit {0}{1\nn}{1\mn}{3.5\nn}{\phantom{ll}}
			\node at (1\mn,2.5/4) {\size{$-j$}};
			\node at (1\mn,5.5/4) {\size{$-i$}};
}}}. \!
\end{equation}

Following the logic of the diagram, we can perform the same operations to obtain that
\begin{equation}
b_{23} b_{34} b_{12} b_{23} \ket{\Omega}^{\otimes n} = \frac{1}{N} \sum_{i,j=0}^{N-1} \zeta^{i^2} c_2^j c_3^{-j} c_1^{i} c_3^{-i} \ket{\Omega}^{\otimes n}.
\end{equation}
By unitarity of the braids, it suffices to show that $\bra{\Omega}^{\otimes n} b_{23} b_{34} b_{12} b_{23} \ket{\Omega}^{\otimes n} = 1$.

Note that the projection onto the ground state yields $\frac{1}{N} \sum_{i,j=0}^{N-1} \zeta^{i^2}\bra{\Omega}^{\otimes n} c_2^jc_3^{-j}c_1^ic_3^{-i} \ket{\Omega}^{\otimes n}=\frac{1}{N} \sum_{i,j=0}^{N-1} \zeta^{i^2}\bra{\Omega}^{\otimes n}c_1^i c_2^jc_3^{-i-j} \ket{\Omega}^{\otimes n}$ by commuting $c_1^i$ past the neutral $c_2^j c_3^{-j}$.
By orthonormality of $c_2^ac_4^b\ket{\Omega}^{\otimes n}$ states, and equivalently, the orthonormality of $c_1^a c_3^b \ket{\Omega}^{\otimes n}$ states, only the terms with $-i-j=0$ survive. Thus, the sum reduces to $\frac{1}{N} \sum_{i=0}^{N-1} \zeta^{i^2}\bra{\Omega}^{\otimes n}c_1^i c_2^{-i}\ket{\Omega}^{\otimes n}$, and this is simply equal to $1$ by lemma \ref{caplemma}.

Thus, it follows by unitarity of the braids that 
\begin{equation}
b_{23} b_{34}b_{12}b_{23}\ket{\Omega}^{\otimes n}=\ket{\Omega}^{\otimes n}.
\end{equation}

In terms of the diagram, for $n=2$, we have
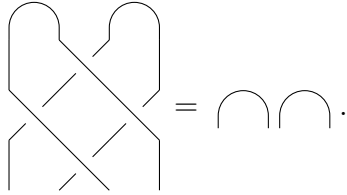
\begin{equation}
\raisebox{-1cm}{	{\tikz{
			\fqudit {0}{6\nn}{1\mn}{.5\nn}{\phantom{ll}}
			\fqudit {-1.333}{6\nn}{1\mn}{.5\nn}{\phantom{ll}}
			\draw (-4/3,4/3)--(-4/3,6/3);
			\draw (-4/3,0/3)--(-4/3,2/3);
			\fbraid{2\mn}{2\nn}{4\mn}{4\nn}
			\fbraid{0\mn}{0}{2\mn}{2\nn}
			\draw (2/3,4/3)--(2/3,6/3);
			\draw (2/3,0/3)--(2/3,2/3);
			\fbraid{-2\mn}{2\nn}{0\mn}{4\nn}
			\fbraid{0\mn}{4\nn}{2\mn}{6\nn}
}}}\, 
=
\raisebox{-.25cm}{
	\tikz{
		\fqudit {0}{0\nn}{1\mn}{.5\nn}{\phantom{ll}}
}}
\raisebox{-.25cm}{
	\tikz{
		\fqudit {0}{0\nn}{1\mn}{.5\nn}{\phantom{ll}}
}}\;.
\end{equation}
\end{proof}

In terms of combinatorial moves, this identity gives us a way to ``slide'' one cap over the other.

\begin{corollary}
	\begin{equation}
    b_{12}b_{23}\ket{\Omega}^{\otimes n}=b_{43} b_{32} \ket{\Omega}^{\otimes n}.
    \end{equation}
\end{corollary}
\begin{proof}
	By taking $b_{34}$ and $b_{23}$ to the right hand side in Proposition \ref{overlap}.
\end{proof}

The above ``slide'' move generalizes to the general result:

\begin{proposition}[General ``Slide'' Move]
\begin{equation} b_{2k,2l-1} b_{2l-1,2l} b_{2k-1,2k} b_{2k,2l-1} \ket{\Omega}^{\otimes n} = \ket{\Omega}^{\otimes n}\end{equation}
for $k<l$ in $\{1,2,\ldots,n\}$.

Note that this result does not generally have a graphical interpretation unless $l=k+1$.
\end{proposition}
\begin{proof}
	Again, by expansion, 
		\begin{equation}b_{2k,2l-1} b_{2l-1,2l}b_{2k-1,2k}b_{2k,2l-1}\ket{\Omega}^{\otimes n}=\frac{\omega}{N} \sum_{i,j=0}^{N-1} c_{2k}^{j} c_{2l-1}^{-j}  b_{2l-1,2l} b_{2k-1,2k} c_{2k}^i c_{2l-1}^{-i}\ket{\Omega}^{\otimes n}. \end{equation}
		
		The same proof as before works in this general case since we can apply the braid intertwining identities and also the twist moves (for braids $b_{2l-1,2l}$ and $b_{2k-1,2k}$), and then apply the axioms to simplify the vacuum expectation value. So we conclude that 
		\begin{equation}b_{2k,2l-1} b_{2l-1,2l}b_{2k-1,2k}b_{2k,2l-1}\ket{\Omega}^{\otimes n}=\ket{\Omega}^{\otimes n}. \end{equation}
\end{proof}

We would also like to be able to ``slip'' one cap in and out of another cap.
\begin{proposition} [``Slip'' Move]
	\label{interlap}
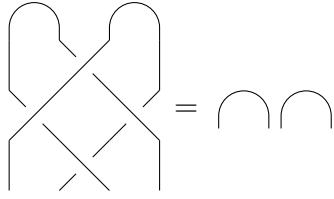
	\begin{equation}\raisebox{-1cm}{
		{\tikz{
				\fqudit {0}{6\nn}{1\mn}{.5\nn}{\phantom{ll}}
				\fqudit {-1.333}{6\nn}{1\mn}{.5\nn}{\phantom{ll}}
				\draw (-4/3,4/3)--(-4/3,6/3);
				\draw (-4/3,0/3)--(-4/3,2/3);
				\fbraid{4\mn}{2\nn}{2\mn}{4\nn}
				\fbraid{0\mn}{0}{2\mn}{2\nn}
				\draw (2/3,4/3)--(2/3,6/3);
				\draw (2/3,0/3)--(2/3,2/3);
				\fbraid{-2\mn}{2\nn}{0\mn}{4\nn}
				\fbraid{2\mn}{4\nn}{0\mn}{6\nn}
	}}}\, 
	=
	\raisebox{-.25cm}{
		\tikz{
			\fqudit {0}{0\nn}{1\mn}{.5\nn}{\phantom{ll}}
	}}
	\raisebox{-.25cm}{
		\tikz{
			\fqudit {0}{0\nn}{1\mn}{.5\nn}{\phantom{ll}}
	}}\;
	\end{equation}
	
	More generally, for $n$ a positive integer not necessarily 1, $$b_{23} b_{34}b_{21}b_{32}\ket{\Omega}^{\otimes n}=\ket{\Omega}^{\otimes n}. $$
\end{proposition}
\begin{proof}
	As demonstrated in the proof of the ``slide'' move, this kind of proof doesn't depend on $n$, so long as $n\geq 2$, so let's specialize to $n=2$ for convenience.
	The previous proposition gave a clear handle on how to manipulate the algebraic computations, so we'll stick with the algebra.
	
	\begin{equation}b_{23} b_{34}b_{21}b_{32}\ket{\Omega}^{\otimes n}=\frac{1}{N} \sum_{i,j=0}^{N-1} c_2^{j} c_3^{-j}  b_{34} b_{21} c_3^i c_2^{-i}\ket{\Omega}^{\otimes n}. \end{equation}

		In terms of a diagram, multiplying the state by $\delta$ (every cap contributes an extra factor of $\sqrt{\delta}$) yields 
		\begin{equation} LHS=\frac{1}{N}\sum_{i,j=0}^{N-1}
	\raisebox{-1cm}{
		{\tikz{
				\fqudit {0}{4\nn}{1\mn}{1.5\nn}{\phantom{ll}}
				\draw (2/3,0/3)--(2/3,2/3);
				\draw (0/3,0/3)--(0/3,2/3);
				\fbraid{0\mn}{2\nn}{-2\mn}{4\nn}
				\node at (-1.5\mn,1/4) {\size{$j$}};
				\node at (-1.2\mn,7/4) {\size{$-i$}};
	}}} 
	\raisebox{-1cm}{
		{\tikz{
				\fqudit {0}{4\nn}{1\mn}{1.5\nn}{\phantom{ll}}
				\draw (2/3,0/3)--(2/3,2/3);
				\draw (0/3,0/3)--(0/3,2/3);
				\fbraid{-2\mn}{2\nn}{0\mn}{4\nn}
				\node at (1\mn,1.75/4) {\size{$-j$}};
				\node at (0.7\mn,6/4) {\size{$i$}};
	}}}\,= 
	\frac{1}{N}\sum_{i,j=0}^{N-1} \!
	\raisebox{-1.15cm}{
		{\tikz{
				\fqudit {0}{4\nn}{1\mn}{0.5\nn}{\phantom{ll}}
				\draw (2/3,-1/3)--(2/3,2/3);
				\draw (0/3,-1/3)--(0/3,2/3);
				\fbraid{0\mn}{2\nn}{-2\mn}{4\nn}
				\node at (-1.5\mn,-1/4) {\size{$j$}};
				\node at (-1.2\mn,1.75/4) {\size{$-i$}};
	}}} \!
	\raisebox{-1cm}{
		{\tikz{
				\fqudit {0}{4\nn}{1\mn}{0.5\nn}{\phantom{ll}}
				\draw (2/3,-1/3)--(2/3,2/3);
				\draw (0/3,-1/3)--(0/3,2/3);
				\fbraid{-2\mn}{2\nn}{0\mn}{4\nn}
				\node at (1\mn,-0.5/4) {\size{$-j$}};
				\node at (0.7\mn,1/4) {\size{$i$}};
	}}},
	\end{equation}
	since the factors of $\zeta^{i^2}$ and $\zeta^{-i^2}$ cancel.
	
	Undoing the twists yields factors of $\omega^{1/2}$ and $\omega^{-1/2}$, respectively, which cancel, so we are left with 
	\begin{equation}
	LHS=
	\frac{1}{N}\sum_{i,j=0}^{N-1} \!
	\raisebox{-0.75cm}{
		{\tikz{
				\fqudit {0\mn}{0.5\nn}{1\mn}{4.5\nn}{\phantom{ll}}
				\node at (-1.2\mn,4\nn) {\size{$-i$}};
				\node at (-1.5\mn,1\nn) {\size{$j$}};
					\fqudit {-4\mn}{0.5\nn}{1\mn}{4.5\nn}{\phantom{ll}}
				\node at (-3.3\mn,3\nn) {\size{$i$}};
				\node at (-3.1\mn,1.5\nn) {\size{$-j$}};
	}}} \;.
	\end{equation}
	
	Converting back to the algebraic form, one has
	\begin{equation}b_{23} b_{34}b_{21}b_{32}\ket{\Omega}^{\otimes n}=\frac{1}{N} \sum_{i,j=0}^{N-1} c_2^{j} c_3^{-j}   c_3^i c_2^{-i}\ket{\Omega}^{\otimes n}. \end{equation}
	
	Note that the $\ket{00}$ component has norm 1, since setting $i=j$ yields the $\ket{00}$ component. Thus, by unitarity of the braid elements, the other basis state projections vanish, so 
		\begin{equation}b_{23} b_{34}b_{21}b_{32}\ket{\Omega}^{\otimes n}=\ket{\Omega}^{\otimes n}\end{equation}
		as desired.
	
\end{proof}

As with the ``slide'' move, there is again an algebraic generalization to braid elements with no graphical interpretation:
\begin{proposition}[General ``Slip'' Move]
	\begin{equation}b_{2k,2l-1} b_{2l-1,2l} b_{2k,2k-1} b_{2l-1,2k} \ket{\Omega}^{\otimes n} = \ket{\Omega}^{\otimes n}\end{equation}
	for $k<l$ in $\{1,2,\ldots, n\}$.
\end{proposition}
\begin{proof}
	By expansion, 
		\begin{equation}b_{2k,2l-1} b_{2l-1,2l}b_{2k,2k-1}b_{2l-1,2k}\ket{\Omega}^{\otimes n}=\frac{1}{N} \sum_{i,j=0}^{N-1} c_{2k}^{j} c_{2l-1}^{-j}  b_{2l-1,2l} b_{2k,2k-1} c_{2l-1}^i c_{2k}^{-i}\ket{\Omega}^{\otimes n},\end{equation}
		and the same proof follows through as before.
\end{proof}

\begin{corollary}
	\begin{equation}b_{21} b_{32} \ket{\Omega}^{\otimes n} = b_{43} b_{32} \ket{\Omega}^{\otimes n} \end{equation}
\end{corollary}
\begin{proof}
	By taking $b_{23}$ and $b_{34}$ to the right hand side in proposition \ref{interlap}.
\end{proof}

\begin{proposition}
	\begin{equation}\raisebox{-0.75cm}{
		{\tikz{
				\fqudit {0}{6\nn}{1\mn}{.5\nn}{\phantom{ll}}
				\fqudit {-1.333}{6\nn}{1\mn}{.5\nn}{\phantom{ll}}
				\draw (-4/3,2/3)--(-4/3,6/3);
				\draw (-2/3,2/3)--(-2/3,4/3);
				\draw (2/3,4/3)--(2/3,6/3);
				\fbraid{-2\mn}{2\nn}{0\mn}{4\nn}
				\fbraid{0\mn}{4\nn}{2\mn}{6\nn}
	}}}\, =
	\raisebox{-0.75cm}{
	{\tikz{
			\fqudit {0}{6\nn}{1\mn}{.5\nn}{\phantom{ll}}
			\fqudit {-1.333}{6\nn}{1\mn}{.5\nn}{\phantom{ll}}
			\draw (-4/3,2/3)--(-4/3,6/3);
			\draw (-2/3,2/3)--(-2/3,4/3);
			\draw (2/3,4/3)--(2/3,6/3);
			\fbraid{0\mn}{2\nn}{-2\mn}{4\nn}
			\fbraid{2\mn}{4\nn}{0\mn}{6\nn}
}}}\, \end{equation}
i.e.
\begin{equation}
b_{34} b_{23} \ket{\Omega}^{\otimes n} = b_{43} b_{32} \ket{\Omega}^{\otimes n}
\end{equation}
\end{proposition}
\begin{proof}
	It suffices to show that $b_{23}b_{34} b_{34} b_{23} \ket{\Omega}^{\otimes n} = \ket{\Omega}^{\otimes n}$, using the fact that $b_{jk} b_{kj} = 1$.
	
	Note that this relation does \textbf{not} follow immediately from the Yang-Baxter-like equation, since the Yang-Baxter-like equation does not know about the vector structure, or even about the behavior of the ground state.
	
	First recall that proposition \ref{overlap} says that the ground state  $\ket{\Omega}^{\otimes n}$ is invariant under a \\ \noindent ``slide'' move via 
	\begin{equation}
	 \ket{\Omega}^{\otimes n} = 	 b_{23} b_{34} b_{12} b_{23} \ket{\Omega}^{\otimes n}
	\end{equation}
	and so we have that
	\begin{equation}
		 b_{32} b_{43} b_{21} b_{32} \ket{\Omega}^{\otimes n} =\ket{\Omega}^{\otimes n}.
	\end{equation}
	
	Thus,
	\begin{align}
	b_{23} b_{34} b_{34} b_{23} \ket{\Omega}^{\otimes n} &= 	b_{23} b_{34} b_{34} b_{23} b_{32} b_{43} b_{21} b_{32} \ket{\Omega}^{\otimes n} \\
    &= b_{23} b_{34} b_{21} b_{32} \ket{\Omega}^{\otimes n}
	\end{align}
	which equals $\ket{\Omega}^{\otimes n}$ by proposition \ref{interlap}, as desired.
		
\end{proof}

Now we prove something quite nontrivial using the above braiding relations in combination.

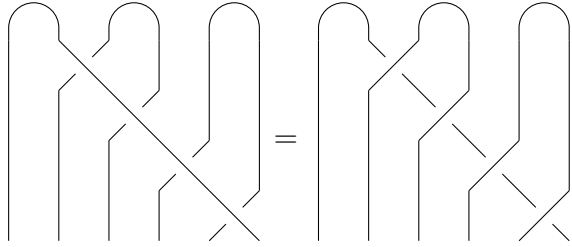
\begin{proposition}
	\begin{equation}\raisebox{-1.25cm}{
		{\tikz{
				\fqudit {0}{6\nn}{1\mn}{.5\nn}{\phantom{ll}}
				\fqudit {-1.333}{6\nn}{1\mn}{.5\nn}{\phantom{ll}}
					\fqudit {1.333}{6\nn}{1\mn}{.5\nn}{\phantom{ll}}
				\draw (-4/3,-2/3)--(-4/3,6/3);
					\draw (-2/3,-2/3)--(-2/3,4/3);
						\draw (0/3,-2/3)--(0/3,2/3);
				\draw (2/3,4/3)--(2/3,6/3);
				\draw (2/3,-2/3)--(2/3,0/3);
					\draw (4/3,2/3)--(4/3,6/3);
						\draw (6/3,0/3)--(6/3,6/3);
				\fbraid{-2\mn}{2\nn}{0\mn}{4\nn}
				\fbraid{0\mn}{4\nn}{2\mn}{6\nn}
					\fbraid{-6\mn}{-2\nn}{-4\mn}{0\nn}
					\fbraid{-4\mn}{0\nn}{-2\mn}{2\nn}
	}}}\, =
	\raisebox{-1.25cm}{
		{\tikz{
				\fqudit {0}{6\nn}{1\mn}{.5\nn}{\phantom{ll}}
				\fqudit {-1.333}{6\nn}{1\mn}{.5\nn}{\phantom{ll}}
					\fqudit {1.333}{6\nn}{1\mn}{.5\nn}{\phantom{ll}}
				\draw (-4/3,-2/3)--(-4/3,6/3);
				\draw (-2/3,-2/3)--(-2/3,4/3);
				\draw (0/3,-2/3)--(0/3,2/3);
				\draw (2/3,4/3)--(2/3,6/3);
				\draw (2/3,-2/3)--(2/3,0/3);
				\draw (4/3,2/3)--(4/3,6/3);
				\draw (6/3,0/3)--(6/3,6/3);
				\fbraid{0\mn}{2\nn}{-2\mn}{4\nn}
				\fbraid{2\mn}{4\nn}{0\mn}{6\nn}
					\fbraid{-4\mn}{-2\nn}{-6\mn}{0\nn}
				\fbraid{-2\mn}{0\nn}{-4\mn}{2\nn}
	}}}\, \end{equation}
	i.e.
	\begin{equation}b_{56} b_{45} b_{34} b_{23} \ket{\Omega}^{\otimes n} = b_{65} b_{54} b_{43} b_{32} \ket{\Omega}^{\otimes n}.\end{equation}
\end{proposition}
\begin{proof}
	Equivalently, we will show that \begin{equation}b_{23} b_{34} b_{45} b_{56} b_{56} b_{45} b_{34} b_{23} \ket{\Omega}^{\otimes n} = \ket{\Omega}^{\otimes n}.\end{equation}

	We first substitute $b_{32} b_{43} b_{21} b_{32} \ket{\Omega}^{\otimes n}$ for $\ket{\Omega}^{\otimes n}$ following Proposition \ref{overlap}. This kills off the $b_{34}$ and $b_{23}$ braids and we are left with 
	\begin{equation}b_{23} b_{34} b_{45} b_{56} b_{56} b_{45}  b_{21} b_{32}  \ket{\Omega}^{\otimes n}.\end{equation}
	
	Now we commute the braids which do not overlap so we get
	\begin{equation}b_{23} b_{34}b_{21} b_{32}  b_{45} b_{56} b_{56} b_{45}   \ket{\Omega}^{\otimes n}.\end{equation}
	
	We now substitute  $b_{54} b_{65} b_{43} b_{54} \ket{\Omega}^{\otimes n}$ for $\ket{\Omega}^{\otimes n}$ to get
	\begin{equation}b_{23} b_{34}b_{21} b_{32}  b_{45} b_{56} b_{43} b_{54}  \ket{\Omega}^{\otimes n}\end{equation}
	upon braid and adjoint braid cancellation. Now we apply the slip move in reverse to get
	\begin{equation}b_{23} b_{34}b_{21} b_{32}  \ket{\Omega}^{\otimes n}\end{equation}
	and then apply the slip move in reverse again to get $\ket{\Omega}^{\otimes n}$, as desired.
\end{proof}

\section{Conclusion}

In this work, we constructed a graphical calculus for multi-qudit computations with the generalized Clifford algebra.  Using purely algebraic methods, we established many graphical and beyond graphical identities of the representation of generalized Clifford algebras considered in \cite{Lin1}, including a novel algebraic proof of a Yang-Baxter equation and a construction of a corresponding braid group representation. Our algebraic proof also enabled a resolution of an open problem in \cite{Cobanera2014} on the construction of self-dual braid group representations for $N$ even. We also derived several new identities for the braid elements, which are key to our proofs. Furthermore, we demonstrated that in many cases, the verification of involved vector identities can be reduced to the combinatorial application of two basic vector identities.

Furthermore, we demonstrated that it is feasible to envision implementing the braid operators for quantum computation, by showing that they are 2-local operators. In fact, as we demonstrated these braid elements are \textit{almost} Clifford gates, for they normalize the generalized Pauli group up to an extra factor $\zeta$.

\section*{Acknowledgments}

I would like to express my gratitude and thanks to  Professor Arthur Jaffe, for helpful discussion and guidance, especially his recommendation to find a more general identity for the braid.

I am grateful to the anonymous referees and editors for their feedback, which have helped improve the manuscript.

I have been supported in the later stages of this work by ARO Grant W911NF-20-1-0082 through the MURI project ``Toward Mathematical Intelligence and Certifiable Automated Reasoning: From Theoretical Foundations to Experimental Realization.''

\bibliographystyle{plainnat}
\bibliography{references}

\end{document}